\documentclass[11pt]{article}

\usepackage{authblk}

\title{On the I/O Complexity of Sequential and Parallel\\ Hybrid Algorithms for Integer Multiplication}

\author{Lorenzo De Stefani\thanks{lorenzo\_destefani@brown.edu}}
\affil{Department of Computer Science, Brown University}

 \usepackage[margin=1.25in]{geometry}

\AtBeginDocument{%
  \providecommand\BibTeX{{%
    \normalfont B\kern-0.5em{\scshape i\kern-0.25em b}\kern-0.8em\TeX}}}

\usepackage{amsmath,amssymb,amsfonts}
\usepackage{algorithmic}
\usepackage{graphicx}
\usepackage{textcomp}
\usepackage{xcolor}
\usepackage{pgf, tikz}
\usetikzlibrary{arrows, automata, petri}
\usepackage{longtable}
\usepackage{todonotes}

\usepackage{amsthm}

\newcommand{\BO}[1]{\mathcal{O}\left(#1\right)}
\newcommand{\BOme}[1]{\Omega\left(#1\right)}
\newcommand{\BT}[1]{\Theta\left( #1\right)}

\newcommand{\io}{I/O}

\newcommand{\alg}{\mathcal{A}}
\newcommand{\nproc}{P}
\newcommand{\dom}{D}

\newcommand{\globalInput}{\mathcal{X}}
\newcommand{\subpInput}{\mathcal{Y}}
\newcommand{\subpOutput}{\mathcal{Z}}

\newcommand{\subpsize}{8M}
\newcommand{\halfsubpsize}{4M}

\newcommand{\msgsize}{B_m}

\DeclareMathOperator*{\argmin}{arg\,min}

\newcommand{\nz}{n_0}

\newcommand{\hmm}{\mathfrak{H}}
\newcommand{\uhmm}{\mathfrak{UH}\left(k,\nz{}\right)}

\newtheorem{theorem}{Theorem}
\newtheorem{lemma}[theorem]{Lemma}
\newtheorem{corollary}[theorem]{Corollary}
\newtheorem{definition}{Definition}

\begin{document}
\begin{titlepage}
\pagenumbering{gobble}
\maketitle

\begin{abstract}
Almost asymptotically tight lower bounds are derived for the Input/Output (\io{}) complexity $IO_\mathcal{A}\left(n,M\right)$ of a general class of hybrid algorithms computing the product of two integers,
each represented with $n$ digits in a given base $s$, in a two-level storage hierarchy with
$M$ words of fast memory, with different digits stored in different memory words. The considered hybrid algorithms combine the  \emph{Toom-Cook-}$k$ (or \emph{Toom}-$k$)  fast integer multiplication approach with  computational complexity $\BT{c_kn^{\log_k \left(2k-1\right)}}$, and ``\emph{standard}''  integer multiplication algorithms which compute $\Omega\left(n^2\right)$ digit multiplications. 
\sloppy We present an  $\Omega\left(\left(n/\max\{M,\nz{}\}\right)^{\log_k \left(2k-1\right)}\left(\max\{1,\nz{}/M\}\right)^2M\right)$ lower bound for the \io{} complexity of a class of ``\emph{uniform, non-stationary}'' hybrid algorithms, where $\nz{}$ denotes the threshold size of sub-problems  which are computed using standard algorithms with algebraic complexity  $\Omega\left(n^2\right)$. As a special case, our result yields an asymptotically tight $\BT{n^2/M}$ lower bound for the \io{} complexity of any standard integer multiplication algorithm. As some sequential hybrid algorithms from this class exhibit \io{} cost within a $\BO{k^2}$ multiplicative term of the corresponding lower bounds, the proposed lower bounds are almost asymptotically tight and indeed tight for constant values of $k$. By extending these results to a distributed memory model with $\nproc{}$ processors, we obtain both memory-dependent and memory-independent \io{} lower bounds for parallel versions of hybrid integer multiplication  algorithms. 
All the lower bounds are derived for the more general class of ``\emph{non-uniform, non-stationary}'' hybrid algorithms that allow recursive calls to have a different structure, even when computing sub-problems with  the same input size, and to use different versions of Toom-$k$. 




\end{abstract}
\end{titlepage}
\pagenumbering{arabic}

\section{Introduction}\label{sec:intro}
Data movement has a critical impact on the performance of computing systems as it affects both execution time and energy utilization. This technological
trend~\cite{patterson2005getting} is destined to continue, as physical
limitations on minimum device size and maximum message speed lead
to unavoidable costs when moving data, whether across the levels of a
hierarchical memory system or between processing elements of a
parallel system~\cite{bilardi1995horizons}. While  communication
requirements of  algorithms have been widely investigated in the literature,
obtaining meaningful lower bounds on the data movement requirement remains
an important and challenging task.

In this paper, we study the \io{} complexity of a class of hybrid algorithms for computing the product of two $n$-digit integer numbers. Such algorithms combine \emph{standard} (or \emph{long}) integer multiplication algorithms that require $\BOme{n^2}$ digit operations, with \emph{fast} algorithms based on the Toom-Cook-$k$ (or Toom-$k$) recursive algorithm~\cite{toom1963complexity} which requires $\Theta\left(c_k n^{\log_k
(2k-1)}\right)$ digit operations where $2k-1$ denotes the number of sub-problems generated at each recursive step, and where $c_k$ is a parameter that is  constant with  respect to $n$
while it grows with $k$.  The considered hybrid algorithms are ``\emph{non-uniform, non-stationary}'' as they allow for different algorithms to be used in different recursion levels (hence, non-stationary), and as they allow recursive calls to have a different structure, even when
they refer to the multiplication of integer values in the same recursive level and of the same input size in terms of the number of digits (hence, non-uniform).  This class models, for example, algorithms that optimize for different input sizes of recursively generated sub-problems. This corresponds to an actual practical scenario: while for large input size the use of fast algorithms such as Toom-Cook's and Karatsuba's (which can be considered as a special case of the former) leads to a reduction of the number of the required digit-wise operations compared to standard $\BOme{n^2}$ algorithms, as the size of the input of the recursively generated sub-problems decreases, the asymptotic benefit of fast algorithms wanes due to the increasing relative impact of the constant multiplicative factor, and standard $\BOme{n^2}$ algorithms become more efficient.

While of actual practical importance, the characterization of hybrid algorithms is challenging due to the intrinsic irregular nature of such algorithms and, hence, the irregular structure of the corresponding Computational Directed Acyclic Graphs (CDAGs). In turn, this considerably  complicates the analysis of the combinatorial properties of the CDAG, which is the foundation of many of \io{} lower bound techniques presented in the literature (e.g.,~\cite{ballard2012graph,jia1981complexity,savage1995extending}).

In~\cite{dest19}, De Stefani proposed a technique that overcomes such challenges and yields asymptotically tight \io{} lower bounds of a general class of hybrid algorithms for matrix multiplication combining \emph{fast} Strassen-like algorithms and standard ``\emph{definition}'' algorithms. In this work, we demonstrate the generality of the approach in~\cite{dest19} by deploying it to the analysis of hybrid integer multiplication algorithms. However, the analysis requires several non-trivial adaptations that correspond to the specific properties of the integer multiplication function and to the combinatorial properties of Toom-Cook algorithms. Further, we extend the technique in~\cite{dest19} to obtain \emph{memory-independent} \io{} lower bounds for a distributed-memory parallel model with $\nproc{}$ processors.  All the \io{} lower bounds for sequential algorithms presented in this work are within at most a $\BO{k_{max}^2}$ multiplicative factor of corresponding upper bounds, where $k_{max}$ denotes the maximum value for any version of Toom-$k$ used in the hybrid scheme, and hence are asymptotically tight whenever $k$ is constant with respect to the size of the input. Ours are the first \io{} lower bounds for hybrid algorithms for integer multiplication.\\

\noindent\textbf{Previous and Related Work:} Karatsuba~\cite{karatsuba1962multiplication} showed that two $n$-digits
integers can be multiplied with $O(n^{\log_2 3})$ digit operations, hence with asymptotically fewer than the $n^2$ operations required by the straightforward implementation of the standard \emph{long}, or ``\emph{Schoolbook}'', integer multiplication algorithms.

Following this effort, the \emph{Toom-Cook}
algorithmic scheme was originally introduced for circuits by Andrei
Toom~\cite{toom1963complexity} and later adapted to software programs by
Stephen Cook~\cite{cook}. The multiplicands $A$
and $B$ are viewed as the result of evaluating, at a suitable point
$z$, two polynomials $p(x)$ and $q(x)$ of degree $k{-}1$ (where $k$ is
a design parameter within the scheme), whose coefficients are integers
directly obtained from the digits of $A$ and $B$,
respectively. The product $AB=p(z)q(z)$ can then be
obtained by evaluating at $z$ the polynomial $r(x)=p(x)q(x)$, of
degree $(2k{-}2)$. 
When the scheme is implemented using the same value $k$ at all levels of recursion the resulting algorithm is commonly referred as
``\emph{Toom}-$k$'' and computes $\Theta\left(c_k n^{\log_k
(2k-1)}\right)$ digit operations, where $c_k$ is a  constant with respect to $n$,
while it grows with $k$. Knuth~\cite[p. 286d]{knuth1981seminum} showed
how, by suitably choosing $k$ as a function of the (sub)problem size,
the Toom-Cook scheme can be implemented to compute $O\left(n
2^{\sqrt{2 \log n}} \log n\right)$ operations.  Asymptotically faster algorithms based on the application of the \emph{Fast Fourier Transform} have been presented in literature. Among these, the \emph{Sch\"onhage-Strassen algorithm}~\cite{schonhage1971schnelle} with complexity $\Theta\left(n \log n \log \log n \right)$, and the \emph{F\"urer algorithm}~\cite{furer2009faster} with complexity $\Theta\left(n \log n 2^{\BO{\log^*n}} \right)$, where $\log^*n$ is the \emph{iterated logarithm}. While the former is used in practice for input integers greater than $2^{2^{17}}$~\cite{garcia2005can}, F\"urer's algorithm and others, asymptotically faster ones (i.e.,~\cite{covanov2019fast,harvey2019faster,harvey2019integer,harvey2016even}), are not practical for reasonable input size.

In practice, Toom-Cook is slower than the standard
algorithm for small numbers.  It becomes competitive for input integers of
intermediate size, up to $2^{2^{17}}$~\cite{garcia2005can}, before the
Sch\"onhage-Strassen becomes actually faster.\\

\io{} complexity has been introduced in the seminal work by Hong and
Kung~\cite{jia1981complexity}; it denotes the number of data
transfers between the two levels of a memory hierarchy with a fast
memory (or \emph{cache}) of $M$ words and a slow memory with an unbounded number of
words. Hong and Kung presented techniques to develop lower bounds for the \io{} complexity of computations modeled by \emph{computational directed acyclic graphs} (CDAGs). The resulting lower bounds
apply to all the possible execution schedules of the given CDAG, including those with recomputation, that is, those for which some vertices of the CDAG are evaluated multiple times. In the same paper, the authors presented \io{} lower bounds for several important algorithms, including standard matrix multiplication and Fast Fourier Transform. The techniques
of~\cite{jia1981complexity} have also been extended to obtain tight
communication bounds for the definition-based matrix multiplication in
some parallel settings~\cite{ballard2012brief,irony2004communication,scquizzato_et_al:LIPIcs:2014:4493}. 

While a general $\left(S+1\right)T = n^2/64$ time ($T$)-space ($S$) tradeoff for the integer multiplication function is known~\cite[Theorem 10. 5.3]{savage97models}, to the best of our knowledge, this work provides the first asymptotically tight \io{} lower bounds for \emph{all} standard \emph{long} integer multiplication algorithms which require $\BOme{n^2}$ digit-wise operations.

In~\cite{BilardiS19}, Bilardi and De Stefani applied the ``\emph{G-flow technique}'' \io{} technique (originally introduced in~\cite{bilardi2017complexity}) to obtain almost asymptotically tight for a class of \emph{uniform, non-stationary} Toom-Cook algorithms. Their technique is based on the analysis of the \emph{Partial Grigoriev's flow} of the integer multiplication function. In this work, we extend their result to a class of hybrid non-uniform, non-stationary algorithms. In~\cite{dest19}, De Stefani provided an analysis of the \io{} complexity of a class of hybrid algorithms for matrix multiplication combining \emph{fast} Strassen-like algorithms and standard ``\emph{definition}'' algorithms. In this work, we build on the general framework introduced in~\cite{dest19} for the analysis of hybrid integer multiplication algorithms.

\noindent\textbf{Our results:}
In our main result, we present an \io{} lower bound for a class $\hmm{}$ of non-uniform, non-stationary hybrid integer multiplication algorithms when executed in a two-level storage hierarchy with $M$ words of fast
memory. Algorithms in $\hmm{}$ combine the Toom-Cook recursive strategy and standard algorithms that require  $\BOme{n^{2}}$ digit operations. Such hybrid algorithms allow recursive calls to have a different structure, even when they refer to the multiplication of matrices in the same recursive level and of the same input size. 
The result in Theorem~\ref{thm:genmatmul} relates the \io{} complexity of algorithms in $\hmm{}$ to the number and the input size of an opportunely selected set of the sub-problems generated by the algorithms themselves. As a special case, Theorem~\ref{thm:corgenmatmul} yields an asymptotically tight $\BT{n^2/M}$ \io{} lower bound for standard integer multiplication algorithms (Corollary~\ref{cor:standardio}). 

\sloppy We specialize the result of Theorem~\ref{thm:genmatmul} to obtain, in Theorem~\ref{thm:corgenmatmul}, an $\Omega\left(\left(\frac{n}{\max\{M,\nz{}\}}\right)^{\log_k (2k-1)}\left(\max\{1,\frac{\nz{}}{M}\}\right)^2M\right)$ lower bound for the \io{} complexity of algorithms in a subclass $\uhmm{}$ of $\hmm{}$ composed by \emph{uniform non-stationary} hybrid algorithms where Toom-$k$ is used in the initial levels of recursion until the size of the generated sub-problems falls below a set threshold $\nz{}$ and those are computed using standard algorithms executing  $\BOme{n^2}$ digit operations. 
The application of the technique introduced in~\cite{dest19} for the analysis of algorithms in~$\hmm{}$ requires a careful combination of  the ``\emph{G-flow}'' \io{} lower bound technique originally introduced by Bilardi and De Stefani for Toom-Cook algorithms in~\cite{BilardiS19}, with an analysis of the combinatorial properties of standard integer multiplication algorithms.

By modifying the approach used to obtain our main results for the hierarchical memory model (or \io{} model), we obtain \io{} lower bounds for distributed-memory parallel models with~$\nproc{}$ processors each equipped with a \emph{local} (i.e., non-shared) memory. For these parallel models, we derive lower bounds for the ``\emph{bandwidth cost}'', that is for  the number of memory words transmitted (either sent or received)  by at least one processor during the execution of the algorithm, and for the ``\emph{latency cost}'', that is for  the number of messages composed of at most $B_m$ memory words transmitted (either sent or received)  by at least one processor during the execution of the algorithm. These results also provide lower bounds to the number of memory words (resp., messages) exchanged on the algorithm's ``\emph{critical path}'' as defined in~\cite{yang1988critical}.  

We obtain both \emph{memory-dependent} lower bounds by assuming a bound on the size of the local memory available to each processor, and \emph{memory-independent} lower bounds by assuming that either the input is originally distributed among the processors in a balanced fashion, or that all the processors contribute in a balanced way to the computational effort of the algorithm. While the most general version of these results characterizes the bandwidth cost of the entire class $\hmm{}$ (Theorem~\ref{thm:genmatmul} for memory-dependent bound, Theorem~\ref{thm:memindbound} for memory-independent bound under balanced input distribution), by specializing these result we characterize the bandwidth cost of parallel algorithms in $\uhmm{}$ (Theorem~\ref{thm:corgenmatmul} for memory-dependent bound, Theorem~\ref{thm:specmemindbound} for memory-independent bound under balanced input distribution, and Corollary~\ref{cor:specmemindboundbalanced2} for memory-independent bound under balanced participation to the computational effort). 

As a special case, these results wield novel lower bound for the bandwidth cost of any standard integer multiplication algorithms (Corollary~\ref{cor:standardio} for memory-dependent bound, Corollary~\ref{cor:mindsta} for memory-independent bound under balanced input distribution, and Theorem~\ref{thm:balancompsta} for memory-independent bound under balanced participation to the computational effort), and an $\BOme{\frac{n}{P^{1/\log_k(2k-1)}}}$ memory-independent bandwidth cost lower bound for uniform Toom-$k$ algorithms assuming a balanced distribution of the input values among the processors.

All the lower bounds presented in this paper allow for the recomputation of intermediate values.   Given the existence of sequential algorithms whose \io{} cost is  within at most a $\BO{k^2}$ multiplicative factor of the lower bounds for sequential hybrid algorithms in Theorem~\ref{thm:genmatmul} and Theorem~\ref{thm:corgenmatmul}, we can conclude that, for values of $k$ which are constant with respect to $n$, these  algorithms are \io{} optimal and our lower bounds are asymptotically tight. Since the (almost) matching algorithms do not recompute any intermediate value, we conclude that using recomputation may reduce the \io{} complexity of the considered classes of hybrid algorithms by at most a $\BO{k^2}$ multiplicative factor.


\section{Preliminaries}\label{sec:preliminaries}
We consider algorithms that compute the product of two integers: $C = A\times B$. We assume the input integers to be expressed as a sequence of $n$ base-$s$ digits in \emph{positional notation}. We further assume that each integer is represented as an unsigned integer with an additional bit to denote the sign. 

For a given integer $A$, we denote its expansion in base $s$ as:
\begin{equation*}
    A = \left(A[n-1],A[n-2],\ldots A[0]\right)_s,
\end{equation*}
where $n$ is the number of digits in the base-$s$ expansion of $A$, and its digits are indexed in order from the least significant digit~$A[0]$ to the most significant digit~$A[n-1]$. Further, we refer to $A[i-1]$ (resp., $A[n-i]$) as the $i$-th \emph{least} (resp., \emph{most}) significant digit of $A$ for $i=1,\ldots,n$. 
With a slight abuse of notation, we use $A$ (resp., $B$) to denote both the value being multiplied and the set of input variables to the algorithm. We refer to the number of digits  of  the base-$s$ expansion of an integer $A$ as its  ``\emph{size}'', and we denote it as $|A|$.

In this work we focus on algorithms whose execution can be modeled
as a \emph{Computational Directed Acyclic Graph} (CDAG) $G=\left(V,E\right)$.
Each vertex $v\in V$ represents either an input value or the
result of a unit-time operation (i.e., an intermediate result or one
of the output values) which is stored using a single memory word. For example, each of the input (resp., output) vertices of $G$ corresponds to one of the $n$ digits of the base $s$ expansion of $A$ and $B$ (resp., $C$). The \emph{directed} edges in $E$ represent data dependencies. That is, a pair of vertices $u,v\in V$ are connected by an edge $(u,v)$ directed from $u$ to $v$ if and only if the value corresponding to $u$ is an operand of the unit time operation which computes the value corresponding to $v$.  A \emph{directed path} connecting vertices $u,v\in V$ is
an ordered sequence of vertices starting with $u$ and ending with $v$, such that there is an edge in $E$ directed from each vertex in the sequence to its successor. 

We say that $G'=\left(V',E'\right)$ is a \emph{sub-CDAG} of $G=\left(V,E\right)$ if
$V'\subseteq V$ and $E' \subseteq E \cap (V'\times V')$. Note that, according to this definition, every CDAG is a sub-CDAG of itself. We say that two sub-CDAGs  $G'=\left(V',E'\right)$ and $G''=\left(V'',E''\right)$ of $G$ are \emph{vertex-disjoint} if $V'\cap V''=\emptyset$. Analogously, two directed paths in $G$ are vertex-disjoint if they do not share any vertex.\\

\noindent\textbf{Dominator sets}
When analyzing the properties of CDAGs we use the concept of \emph{dominator set} originally introduced in~\cite{jia1981complexity}. In this work, we use the following, slightly modified, definition:

\begin{definition}[\cite{BilardiS19}]\label{def:dominator}
Given a CDAG $G=(V,E)$, let $I\subset V$ denote the set of its input
vertices. A set $\dom \subseteq V$ is a \emph{dominator set} for $V'\subseteq V$ with respect to $I'\subseteq I$ if every path from a vertex in $I'$ to a vertex in
$V'$ contains at least a vertex of $\dom$. 
When $I'=I$, $\dom$ is simply referred as ``a dominator set for $V'$''.
\end{definition}

Let $G^{\mathcal{A}_n}$ denote the CDAG corresponding to an unspecified straight-line algorithm for computing the product of two input integers $A,B$  whose base-$s$ expression has up to $n$ digits. The following lemma provides a lower bound to the size of any dominator set for selected sets of the output vertices of $G^{\mathcal{A}_n}$:

\begin{lemma}\label{cor:flowcor}
Given $G^{\mathcal{A}_n}$, let $A$ and $B$ denote the two input factor integers to be multiplied expressed as $n$-digit numbers in base $s$, and let $C$ denote the product integer whose base-$s$ expansion has up to $2n$ digits.
Let $X_A$ (resp., $X_B$) denote the set of input vertices of $G^{\mathcal{A}_n}$ associated with the variables corresponding to the $\lfloor n/2\rfloor$ least significant digits of $A$ (resp., $B$). Further, let $Y$ denote the set composed by the $\lfloor n/2\rfloor$ variables corresponding to the digits of $C$ from the $\lfloor n/2\rfloor+1$-th least significant to the $n$-th least significant.
For any $X'\subseteq X_A$ (resp., $X'\subseteq X_B$) and any $Y'\subseteq Y$ we have that any dominator $D$ of $Y'$ with respect to $X'$ must be such that
\begin{equation*}
|D| \geq |X'||Y'|/n.
\end{equation*}
\end{lemma}

The property stated in Lemma~\ref{cor:flowcor} was originally introduced by Bilardi and De Stefani in~\cite{BilardiS19}. Its proof combines an analysis of the \emph{Partial Grigoriev's flow} of the integer multiplication function~\cite[Corollary 5]{BilardiS19}, and its relation with the size of dominator sets of selected vertices of $G^{\mathcal{A}_n}$ with respect to opportunely chosen subsets of the input vertices~\cite[Lemma 1]{BilardiS19}. We refer the reader to~\cite{BilardiS19} for the detailed proofs.

Remarkably, Lemma~\ref{cor:flowcor} captures a property that is common to \emph{all CDAGs} corresponding to \emph{any} integer multiplication integer. This derives from the analysis of the Partial Grigoriev's Flow which is a property of the integer multiplication function itself, rather than of a particular algorithm used to compute it.  Such generality makes this result particularly helpful for the analysis of hybrid algorithms due to their intrinsic heterogeneity. \\

\noindent\textbf{\io{} model and machine models:}
We assume that sequential computations are executed on a system with a two-level memory hierarchy, consisting of a fast memory or \emph{cache} of size $M$ (measured in memory words) and a \emph{slow memory} of unlimited size. An operation can be executed only if all its operands are in cache. 

We assume both the input integers and intermediate results to be stored in memory expressed as their base-$s$ expansion, with $s\in \mathbb{N}^{+}$, and with $s$ being at most equal to the maximum value which can be maintained in a single memory word plus one. (That is, if a memory word can fit 32 bits, we have $2\leq s \leq 2^{32}-1$.) In particular, we assume each digit in the base-$s$ expansion of a value to be stored in a different memory word.
We assume that the processor is equipped with digit-wise product and algebraic sum elementary operations. Further, we assume that the processor is equipped with operations for producing the most and least significant digits of an integer in base $s$. Unless explicitly stated otherwise, when referring to the ``\emph{digits}'' of integers, we mean the digits of their expansion in the base chosen for their representation in memory (i.e., $s$).

Data can be moved from the slow memory
to the cache by \texttt{read} operations, and, in the other direction,
by \texttt{write} operations. These operations are called
\emph{\io{} operations}. We assume that the input data is stored in the slow
memory at the beginning of the computation. The evaluation of a CDAG in this model can be analyzed using the ``\emph{red-blue pebble game}''~\cite{jia1981complexity}. The number of \io{} operations executed when evaluating a CDAG depends on the ``\emph{computational schedule}'', that is, it depends on the order in which vertices are evaluated and on which values are kept in/discarded from the cache. 
The \emph{\io{} complexity} of a CDAG $G$, denoted as $IO_{M}(G)$, is defined as the minimum number of \io{} operations over all possible computational schedules.
We further consider a generalization of this model known as the ``\emph{External Memory Model}'' by Aggarwal and Vitter~\cite{Aggarwal:1988:ICS:48529.48535},
where $B\geq 1$ values can be moved between cache and consecutive
slow memory locations with a single \io{} operation. 

For any given algorithm $\mathcal{A}$, in this work we only consider ``\emph{parsimonious execution schedules}'', that is schedules such that: (i) each time an intermediate result (excluding the output digits of $C$) is computed, the result is then used to compute at least one of the values of which it is an operand before being removed from the memory (either the cache or slow memory); and (ii) any time an intermediate result is read from the slow to the cache memory, such value is then used to compute at least one of the values of which it is an operand before being removed from the memory or moved back to slow memory using a \texttt{write} \io{} operation. Clearly, any non-parsimonious schedule $\mathcal{C}$ can be reduced to a parsimonious schedule $\mathcal{C}'$ by removing all the steps which violate the definition of parsimonious computation. $\mathcal{C}'$ has therefore less computational and \io{} operations than $\mathcal{C}$. Hence, restricting the analysis to parsimonious computations leads to no loss of generality.

We also consider parallel models with distributed memory where $\nproc{}$ processors are connected by a network and can exchange messages with up to $\msgsize{}$ memory words by means of point-to-point communications. We assume that each processor is equipped with a \emph{local} (non-shared) memory. For these parallel models, we derive lower bounds for the ``\emph{bandwidth cost}'' (corresponding to the case $\msgsize{}=1$), that is for the number of memory words transmitted (either sent or received) by at least one processor during the execution of the algorithm,  and for the \emph{latency cost}, that is for  the number of messages composed of at most $B_m$ memory words transmitted (either sent or received)  by at least one processor during the execution of the algorithm. 
The notion of ``\emph{parsimonious execution schedules}'' straightforwardly extends to these parallel models.

\section{Hybrid integer multiplication algorithms}\label{sec:algdef}
We consider a family of hybrid integer multiplication algorithms obtained by combining the two following classes of algorithms:
 \vspace{2mm}

\noindent\textbf{Standard integer multiplication algorithms:} This class includes all the integer  multiplication algorithms which, given the input factor integers $A,B$ which have up to $n$ digits in their respective base-$s$ expansion, satisfy the following properties: 
    \begin{itemize}
        \item The $n^2$ \emph{elementary products}  $A[i]B[j]$, for $i,j=0,1,\ldots, n-1$, are \emph{directly computed};
        \item Let $C=A\times B$. For $i=0,1,\ldots,2n-2$, each of the $C[i]$'s is computed by  summing the values of the elementary products $A[j]B[k]$, for all integer values $0\leq j,k\leq n-1$ such that $j+k=i$, through a \emph{summation tree} by additions and subtractions only and by summing the eventual carry-over value;
        \item The evaluations of the $C[i]$'s are \emph{independent of each other}. That is, internal vertex sets of the summation trees of all the $C[i]$'s are \emph{disjoint from each other}. 
    \end{itemize}   
    Algorithms in this class have computational complexity $\Omega\left({n^2}\right)$.
    This class includes, among others, the \emph{long} or \emph{Schoolbook} algorithm, the \emph{Grid-method multiplication}, and the \emph{Lattice multiplication}.
 \vspace{2mm}

  \noindent\textbf{Toom-Cook algorithms:} This class includes algorithms following the Toom-Cook paradigm~\cite{toom1963complexity}. Algorithms in this class are also referred as \emph{Toom}-$k$ algorithms, where $k>0$ denotes the number of parts in which the digits of the input operands are divided at every recursive step to generate $2k-1$ sub-problems. While variants of the Toom-Cook paradigm for which $k$ can be a non-integer number (e.g., $k=1.5$, $k=2.5$) have been presented in the literature, in this work we consider only versions of the Toom-Cook paradigm for which $k\geq 2$ is an integer. Algorithms in this class follow five common steps: 
    \begin{enumerate}
        \item{\textbf{Splitting:}} The digits of the base-$s$ expansion of $A$ (resp., $B$) are split into $k$ blocks of $\lceil n/k \rceil$ successive digits:
        \begin{align*}
        A_i &= \left(A[\min\{n,(i+1)\lceil n/k \rceil\}-1]\ldots A[i\lceil n/k \rceil]\right)_s,\\
        B_i &= \left(B[\min\{n,(i+1)\lceil n/k \rceil\}-1]\ldots B[i\lceil n/k \rceil]\right)_s,
        \end{align*}
        for $0\leq i\leq k-1$.
        The $A_i$'s and $B_i$'s are then used as coefficients to construct the $k-1$-degree polynomials, $p(x)= \sum_{i=0}^{k-1} A_{i}x^{i-1}$ and $q(x)= \sum_{i=0}^{k-1} B_{i}x^{i-1}$, with the property that $p(s^{\lceil n/k \rceil}) = A$ and $q(s^{\lceil n/k \rceil}) = B$.
        \item \textbf{Evaluation:} In this step the algorithm evaluates the polynomials $p(x)$ and $q(x)$ for the $2k-1$ chosen \emph{distinct evaluation points} $x_0,x_1,\ldots,x_{2k-2}$. Let $\mathbf{x}$ denote the vector such that $\mathbf{x}[i]=x_i$ for $0\leq i\leq 2k-2$, and
let $\textbf{V}\left(\textbf{x}\right)$ denote the \emph{Vandermonde matrix}, such that the the value of the entry in the $i$-th row and $j$-th column is given by:
\begin{equation*}
\left(\textbf{V}\left(\textbf{x}\right)\right)_{i,j}= x_i^j,\ for\  0\leq i,j \leq 2k-2.
\end{equation*}
Let $p(\mathbf{x})$ (resp., $q(\mathbf{x})$) denote the column vector such that $p(\mathbf{x})[i] = p(x_i)$ (resp., $q(\mathbf{x})[i] = q(x_i)$) for $0\leq i \leq 2k-2$.
Finally, let $\mathbf{A}$ (resp., $\mathbf{B}$) denote the column vectors such that $\mathbf{A}[i]= A_i$ for $0\leq i\leq k-1$, and $\mathbf{A}[i]= 0$ for $k\leq i\leq 2k-2$ (resp., $\mathbf{B}[i]= B_i$ for $0\leq i\leq k-1$, and $\mathbf{B}[i]= 0$ for $k\leq i\leq 2k-2$).

The values of the $p(x_i)$'s and $q(x_i)$'s are computed as: $
p(\mathbf{x}) =  \mathbf{V}(\mathbf{x})\mathbf{A}$, $q (\mathbf{x})=  \mathbf{V}(\mathbf{x})\mathbf{B}$. 
While any choice of $2k-1$ distinct evaluations point for the Vandermonde matrix is admissible, small integer values (eg., $0,1,-1,2,\ldots$) are generally chosen in order to reduce the computational cost of the algorithm. 
\item \textbf{Recursive multiplications:}
 Let $r(\cdot) = p(\cdot)q(\cdot)$ denote the \emph{product polynomial} with degree $2k-1$. In this step, the algorithm computes the integer products $r(x_i)=p(x_i)q(x_i)$ for $0\leq i \leq 2k-2$, which are \emph{smaller instances} (for opportunely chosen evaluation points) of the original integer multiplication problem.
    \item \textbf{Interpolation:}
        The coefficients $\mathbf{r}[0], \mathbf{r}[1],\ldots,\mathbf{r}[2k-2]$ of the product polynomial $\mathbf{r}$ are computed as $
            \mathbf{r}= \mathbf{V}(\mathbf{x})^{-1}r(\mathbf{x})
        $.
        As the $2k-1$ evaluations points are chosen to be distinct, due to the properties of Vandermonde matrices~\cite{macon1958inverses}, $\mathbf{V}(\mathbf{x})$, is guaranteed to be invertible. 
        \item \textbf{Recomposition:} Compute $C = \sum_{i=0}^{2k-2}r_is^{i{\lceil n/k \rceil}}$ where multiplications by powers of $s^{\lceil n/k \rceil}$ correspond to shifting the $r_i$ by $i\lceil n/k \rceil$ base-$s$ digits to the left.
    \end{enumerate}
    Algorithms in this class compute $\Theta\left(c_k n^{\log_k
(2k-1)}\right)$ digit operations, where $c_k$ is a constant with respect to $n$,
while it grows with $k$. Sometimes (e.g., in Karatsuba's algorithm~\cite{karatsuba1962multiplication}), the value of $p(x_{2k-2})$ (resp., $q(x_{2k-2})$) is replaced with the coefficient $A_{k-1}$ (resp., $B_{k-1}$) of $x^{k-1}$ in $p(\cdot)$ (resp., $q(\cdot)$). Correspondingly, $r(x_{2k-2})=p(x_{2k-2})q(x_{2k-2})$ equals the coefficient of $x^{2k-2}$ in $r(\cdot)$. In this case, the matrices $\mathbf{V(x)}$ and $\mathbf{V(x)}^{-1}$ can be easily modified. 
This requires minor changes in the algorithm~\cite{bodrato2007towards},\cite{cook}. In order to simplify the presentation, in our lower bound analysis, we assume that this special variation is not used. Our proofs can, however, be straightforwardly  extended to include this case, although the required modifications are more tedious than interesting.
    For further details on the Toom-Cook algorithm and its variations, we refer the reader to~\cite{BilardiS19,bodrato2007towards}.
\vspace{3mm}    

Remarkably, the only properties of relevance for the analysis of the \io{} complexity of algorithms in these classes are those used in the characterization of the classes themselves. While implementation details of these algorithms (e.g., the choice of the evaluation points used for Toom-Cook algorithms, the choice on  carry-over management in summations) may impact the execution time, they will not affect our \io{} lower bound analysis. 
This is particularly desirable when analyzing hybrid algorithms due to their intrinsic heterogeneous nature.

In this work, we consider a general class of ``\emph{non-uniform, non-stationary}'' hybrid integer multiplication algorithms, which allow mixing of schemes from the fast Toom-Cook class with algorithms from the standard class. Given an algorithm $\mathcal{A}$ let $P$ denote the problem corresponding to the computation of the product of the input factor integers $A$ and $B$. Consider an ``\emph{instruction function}''  $f_\mathcal{A}(P)$, which, given as input $P$ returns either (a) indication regarding the algorithm from the standard class which is to be used to compute $P$, or (b) indication regarding the value $k$ for the Toom-$k$ algorithm to be used to recursively generate $2k-1$ sub-problems $P_1,P_2,\ldots,P_{2k-1}$ and the instruction functions $f_\mathcal{A}(P_i)$ for each of the $2k-1$ sub-problems, where $k\in \mathbb{N}$ and $k\geq 2$. We denote as $\hmm{}$ the class of non-uniform, non-stationary algorithms which can be characterized by means of such instruction functions and such that 
if Toom-$k$ is used to compute a (sub)problem with input size $n$, the values $k$ and evaluation points are chosen in such a way that each of the $2k-1$ generated sub-problems has input size at most $2n/k< n$ (i.e., the input size of the sub-problems is strictly decreasing).
Algorithms in $\hmm{}$ allow recursive calls to have a different structure, even when
they refer to the multiplication of matrices in the same recursive level (i.e., non-uniform). E.g., some of the sub-problems with the same size may be  computed using algorithms form the standard class, while others may be computed using recursive algorithms from the fast class, even using different versions of Toom-$k$ (i.e., with different values of $k$). 



We also consider a sub-class $\uhmm{}$ of $\hmm{}$ constituted by ``\emph{uniform, non-stationary}'' hybrid algorithms for which the same version of Toom-$k$ (i.e., using the same value $k$) is used in the first $\ell$ recursion levels, before switching to using an algorithm from the standard class for computing sub-problems of size smaller than or equal to $\nz{}$. Algorithms in this class are \emph{uniform}, i.e., sub-problems of the same size are all either recursively computed using a scheme form the fast class, or are all computed using algorithms from the standard class.

\section{The CDAG of algorithms in $\hmm{}$}\label{sec:CDAG}
Let $G^{\mathcal{A}}= (V_\mathcal{A},E_\mathcal{A})$ denote the CDAG that corresponds to an algorithm $\mathcal{A}\in\hmm{}$ used for multiplying  input  integers $A,B$ of size $n$. Rather than considering a single algorithm, we aim to characterize the properties of the CDAGs corresponding to \emph{any} algorithm in the class $\hmm{}$ that will be relevant for the \io{} analysis. Such analysis appears challenging due to the variety and irregularity of algorithms in $\hmm{}$. In this section, we show a general template for the construction of  $G^{\mathcal{A}}$, and we identify some of its properties which, crucially, hold \emph{regardless} of the implementation details specific of $\mathcal{A}$ and, hence, of  $G^{\mathcal{A}}$.\\

\noindent\textbf{Construction:}
$G^{\mathcal{A}}$ is obtained using a recursive construction that mirrors the recursive structure of the algorithm itself. Let $P$ denote the entire integer multiplication problem computed by $\mathcal{A}$. 
Consider the case for which, according to the \emph{instruction function} $f_{\mathcal{A}}(P)$, $P$ is to be computed using an algorithm from the standard class. As we do not fix a specific algorithm, we do not correspondingly have a fixed CDAG. The only feature of interest for the analysis is that, in this case, the CDAG $G^{\mathcal{A}}$  corresponds to the execution of an algorithm from the standard class for input integers of size $n$.

Consider instead the case for which, according to $f_{\mathcal{A}}(P)$, $P$ is to be computed using a Toom-$k$ algorithm by generating the $2k-1$ sub-problems $P_1,P_2,\ldots,P_{2k-1}$ each associated with an \emph{instruction function}. The sub-CDAGs of $G^{\mathcal{A}}$ corresponding to each of the $2k-1$ sub-problems $P_i$, denoted as $G^{\mathcal{A}_{P_i}}$ are constructed according to $f_{\mathcal{A}}(P_i)$, following recursively the previous steps.
$G^{\mathcal{A}}$ can then be constructed by composing the $2k-1$ sub-CDAGs $G^{\mathcal{A}_{P_i}}$ by means of an \emph{encoder sub-CDAG} $Enc_{k,n}$  which is used to connect the input vertices of $G^{\mathcal{A}}$, which correspond to the values of the input integers $A$ (resp., $B$) to the appropriate input vertices of the $2k-1$ sub-CDAGs $G^{\mathcal{A}_{P_i}}$; the output vertices of the sub-CDAGs $G^{\mathcal{A}_{P_i}}$ (which correspond to the outputs of the $2k-1$ sub-products) are connected to the appropriate output vertices of the entire $G^{\mathcal{A}}$ CDAG by means of a \emph{decoder} sub-CDAG $Dec$.  
The input vertices of the sub-CDAG $Enc_{k,n}$ correspond to the $n$ digits in the base-$s$ expression of the input integers $A$ and $B$, while the output vertices of $Enc_{k,n}$ correspond to the digits of the inputs of the $2k-1$ sub-problems generated by the recursive strategy of $\mathcal{A}$. 
We present an example of such recursive construction in Figure~\ref{fig:strarec}.\\

  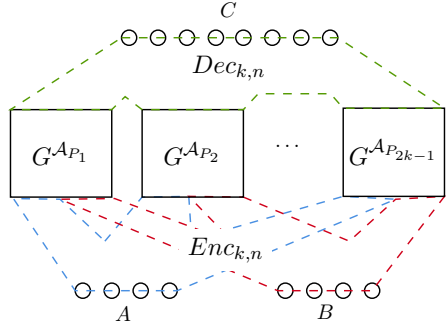
\begin{figure}[bt]
  \centering
     \resizebox{0.4\linewidth}{!}{

\tikzset{every picture/.style={line width=0.75pt}} 

\begin{tikzpicture}[x=0.75pt,y=0.75pt,yscale=-1,xscale=1]

\draw   (80,236.25) .. controls (80,233.35) and (82.35,231) .. (85.25,231) .. controls (88.15,231) and (90.5,233.35) .. (90.5,236.25) .. controls (90.5,239.15) and (88.15,241.5) .. (85.25,241.5) .. controls (82.35,241.5) and (80,239.15) .. (80,236.25) -- cycle ;
\draw   (100,236.25) .. controls (100,233.35) and (102.35,231) .. (105.25,231) .. controls (108.15,231) and (110.5,233.35) .. (110.5,236.25) .. controls (110.5,239.15) and (108.15,241.5) .. (105.25,241.5) .. controls (102.35,241.5) and (100,239.15) .. (100,236.25) -- cycle ;
\draw   (120,236.25) .. controls (120,233.35) and (122.35,231) .. (125.25,231) .. controls (128.15,231) and (130.5,233.35) .. (130.5,236.25) .. controls (130.5,239.15) and (128.15,241.5) .. (125.25,241.5) .. controls (122.35,241.5) and (120,239.15) .. (120,236.25) -- cycle ;
\draw   (140,236.25) .. controls (140,233.35) and (142.35,231) .. (145.25,231) .. controls (148.15,231) and (150.5,233.35) .. (150.5,236.25) .. controls (150.5,239.15) and (148.15,241.5) .. (145.25,241.5) .. controls (142.35,241.5) and (140,239.15) .. (140,236.25) -- cycle ;
\draw   (220,236.25) .. controls (220,233.35) and (222.35,231) .. (225.25,231) .. controls (228.15,231) and (230.5,233.35) .. (230.5,236.25) .. controls (230.5,239.15) and (228.15,241.5) .. (225.25,241.5) .. controls (222.35,241.5) and (220,239.15) .. (220,236.25) -- cycle ;
\draw   (240,236.25) .. controls (240,233.35) and (242.35,231) .. (245.25,231) .. controls (248.15,231) and (250.5,233.35) .. (250.5,236.25) .. controls (250.5,239.15) and (248.15,241.5) .. (245.25,241.5) .. controls (242.35,241.5) and (240,239.15) .. (240,236.25) -- cycle ;
\draw   (260,236.25) .. controls (260,233.35) and (262.35,231) .. (265.25,231) .. controls (268.15,231) and (270.5,233.35) .. (270.5,236.25) .. controls (270.5,239.15) and (268.15,241.5) .. (265.25,241.5) .. controls (262.35,241.5) and (260,239.15) .. (260,236.25) -- cycle ;
\draw   (280,236.25) .. controls (280,233.35) and (282.35,231) .. (285.25,231) .. controls (288.15,231) and (290.5,233.35) .. (290.5,236.25) .. controls (290.5,239.15) and (288.15,241.5) .. (285.25,241.5) .. controls (282.35,241.5) and (280,239.15) .. (280,236.25) -- cycle ;
\draw   (112,61.25) .. controls (112,58.35) and (114.35,56) .. (117.25,56) .. controls (120.15,56) and (122.5,58.35) .. (122.5,61.25) .. controls (122.5,64.15) and (120.15,66.5) .. (117.25,66.5) .. controls (114.35,66.5) and (112,64.15) .. (112,61.25) -- cycle ;
\draw   (132,61.25) .. controls (132,58.35) and (134.35,56) .. (137.25,56) .. controls (140.15,56) and (142.5,58.35) .. (142.5,61.25) .. controls (142.5,64.15) and (140.15,66.5) .. (137.25,66.5) .. controls (134.35,66.5) and (132,64.15) .. (132,61.25) -- cycle ;
\draw   (152,61.25) .. controls (152,58.35) and (154.35,56) .. (157.25,56) .. controls (160.15,56) and (162.5,58.35) .. (162.5,61.25) .. controls (162.5,64.15) and (160.15,66.5) .. (157.25,66.5) .. controls (154.35,66.5) and (152,64.15) .. (152,61.25) -- cycle ;
\draw   (172,61.25) .. controls (172,58.35) and (174.35,56) .. (177.25,56) .. controls (180.15,56) and (182.5,58.35) .. (182.5,61.25) .. controls (182.5,64.15) and (180.15,66.5) .. (177.25,66.5) .. controls (174.35,66.5) and (172,64.15) .. (172,61.25) -- cycle ;
\draw   (191,61.25) .. controls (191,58.35) and (193.35,56) .. (196.25,56) .. controls (199.15,56) and (201.5,58.35) .. (201.5,61.25) .. controls (201.5,64.15) and (199.15,66.5) .. (196.25,66.5) .. controls (193.35,66.5) and (191,64.15) .. (191,61.25) -- cycle ;
\draw   (211,61.25) .. controls (211,58.35) and (213.35,56) .. (216.25,56) .. controls (219.15,56) and (221.5,58.35) .. (221.5,61.25) .. controls (221.5,64.15) and (219.15,66.5) .. (216.25,66.5) .. controls (213.35,66.5) and (211,64.15) .. (211,61.25) -- cycle ;
\draw   (231,61.25) .. controls (231,58.35) and (233.35,56) .. (236.25,56) .. controls (239.15,56) and (241.5,58.35) .. (241.5,61.25) .. controls (241.5,64.15) and (239.15,66.5) .. (236.25,66.5) .. controls (233.35,66.5) and (231,64.15) .. (231,61.25) -- cycle ;
\draw   (251,61.25) .. controls (251,58.35) and (253.35,56) .. (256.25,56) .. controls (259.15,56) and (261.5,58.35) .. (261.5,61.25) .. controls (261.5,64.15) and (259.15,66.5) .. (256.25,66.5) .. controls (253.35,66.5) and (251,64.15) .. (251,61.25) -- cycle ;
\draw  [color={rgb, 255:red, 74; green, 144; blue, 226 }  ,draw opacity=1 ][dash pattern={on 4.5pt off 4.5pt}] (100.5,202) -- (126.25,171.5) -- (158.5,171) -- (160.5,201) -- (265.25,171) -- (301.5,173) -- (150.5,236.25) -- (80,236.25) -- (37.25,172.5) -- (69.5,173) -- cycle ;
\draw  [color={rgb, 255:red, 208; green, 2; blue, 27 }  ,draw opacity=1 ][dash pattern={on 4.5pt off 4.5pt}] (196.25,171.5) -- (269.5,202) -- (301.5,173) -- (335.25,171) -- (285.25,236.25) -- (220,236.25) -- (69.5,173) -- (107.25,172.5) -- (189.5,203) -- (158.5,171) -- cycle ;
\draw   (35.25,111) -- (105.25,111) -- (105.25,171.5) -- (35.25,171.5) -- cycle ;
\draw   (265.25,110) -- (335.25,110) -- (335.25,171) -- (265.25,171) -- cycle ;
\draw   (126.25,111) -- (196.25,111) -- (196.25,171.5) -- (126.25,171.5) -- cycle ;
\draw  [color={rgb, 255:red, 85; green, 155; blue, 4 }  ,draw opacity=1 ][dash pattern={on 4.5pt off 4.5pt}] (112,61.25) -- (261.5,61.25) -- (335.25,110) -- (265.25,110) -- (256.5,100) -- (206.5,100) -- (196.25,111) -- (126.25,111) -- (114.5,102) -- (105.25,111) -- (35.25,111) -- cycle ;

\draw (113,252) node  [align=left] {$\displaystyle A$};
\draw (253,251) node  [align=left] {$\displaystyle B$};
\draw (186,42) node  [align=left] {$\displaystyle C$};
\draw (186,82) node [scale=1.2] [align=left] {$\displaystyle Dec_{k,n}$};
\draw (184,208) node [scale=1.2] [align=left][fill=white] {$\displaystyle Enc_{k,n}$};
\draw (228,136) node  [align=left] {$\displaystyle \dotsc $};
\draw (70.25,141.25) node [scale=1.2]  {$G^{\mathcal{A}_{{P}_{1}}}$};
\draw (161.25,141.25) node [scale=1.2]  {$G^{\mathcal{A}_{{P}_{2}}}$};
\draw (300.25,140.5) node [scale=1.2]  {$G^{\mathcal{A}_{{P}_{2k-1}}}$};

\end{tikzpicture}

     }
     \caption{Simplified recursive construction of $G^{\mathcal{A}}$. The blue (resp., red) dotted lines represent the combination of the digits of $A$ (resp., $B$) as discussed in the \emph{Evaluation} step of the Toom-$k$ algorithm to obtain the input values to the $2k-1$ generated sub-problems. The green dotted lines represent the computation of the product $C$ by composing the results of the sub-problems as discussed in the \emph{Interpolation} step of  Toom-$k$.
    }
  \label{fig:strarec}
\end{figure} 
\noindent\textbf{Properties of $G^{\mathcal{A}}$:}
While the exact internal structure of $G^{\mathcal{A}}$, and, in particular, the structure of encoder and decoder sub-CDAGs depends on the exact implementation details of the specific Toom-Cook algorithm being used by $\mathcal{A}$ (e.g., the choice of evaluation points, how to deal with carryover), \emph{all} versions share the following properties, originally discussed in~\cite{BilardiS19}, which are of great importance for the remainder of the analysis. We refer the reader to~\cite{BilardiS19} for the proofs.

\begin{lemma}{\cite[Lemma 3.1]{BilardiS19}}\label{lem:distinct}
Let $\mathcal{A}\in \hmm{}$ and denote an algorithm that recursively computes the product of two input integers $A$ and $B$, whose base-$s$ expansions have $n$ digits each. Let $P_1$ and $P_2$ be any two sub-problems generated by $\mathcal{A}$ using  algorithms from the Toom-Cook class, such that $P_2$ is not recursively generated by $P_1$ and vice versa. Then, the sub-CDAG of $G^{\mathcal{A}}$ corresponding to $P_1$ and the sub-CDAG corresponding to $P_2$ are vertex disjoint. 
\end{lemma}

The following lemma captures a connectivity property  shared by the encoder sub-CDAGs for \emph{all} Toom-$k$ algorithms:

\begin{lemma}[{\cite[Lemma 3.2]{BilardiS19}}]\label{lem:conneconder}
For $k\geq 2$, given the encoder $Enc_{k,n}$, let $I'_A$  (resp., $I'_B$) denote the subset of the input vertices that correspond to the $\lceil n/k \rceil$ least significant digits in the base-$s$ expansion of $A$ (resp., $B$). Further, for $0\leq i \leq 2k-2$, let $O^{(i)}_A$ (resp., $O^{(i)}_B$) denote the subset of the output vertices of $Enc_{k,n}$ that correspond to the $\lceil n/k \rceil$ least significant digits of the input of the $i$-th sub-problem obtained from the values of $A$ (resp., $B$). 

For any $0\leq i \leq 2k-2$, consider $Y\subseteq O^{(i)}_A$ (resp., $\subseteq O^{(i)}_B$). We then have that there exists $X_A\subseteq I'_A$ (resp., $X_B\subseteq I'_B$) with $|X_A| = |Y|$ (resp., $|X_b| = |Y|$), such that there exist $|Y|$ vertex disjoint paths connecting the vertices in $X_A$ (resp., $X_B$) to the vertices in $Y$.
\end{lemma}

\section{Maximal sub-problems and their properties}\label{sec:msp}
 For an algorithm $\mathcal{A}\in \hmm{}$, let $P'$ denote a sub-problem generated by $\mathcal{A}$. 
In our presentation we consider the entire computation of the product $C=A\times B$ as an \emph{improper} sup-problem generated by $\mathcal{A}$. 
For a given sub-problem $P'$ generated by an algorithm $\mathcal{A}\in \hmm{}$, let  $P'_1, P'_2,\ldots,P'_i$ be the sequence of sub-problems  generated by $\mathcal{A}$ such that $P'_{j+1}$ was recursively generated to compute $P'_j$ for $j=1,2,\ldots,i-1$, and such that $P'$ was recursively generated to compute $P'_i$. We refer to the sub-problems in such a sequence as the \emph{ancestor sub-problems} of $P'$. Clearly, if $P'$ is the entire integer multiplication problem being computed by $\mathcal{A}$, it has no ancestor sub-problems.

To study the \io{} complexity of algorithms in $\hmm{}$, we focus on the analysis of a particular group of sub-problems.

\begin{definition}[Integer Multiplication Maximal Sub-Problems (MSP)]\label{def:maxsubp}
Let $\mathcal{A} \in \hmm{}$ be an algorithm used to multiply two integers $A,B$ whose base-$s$ expansion has at most $n$ digits. If $n< n'$ we say that $\mathcal{A}$ does not generate any $n'$-Maximal Sub-Problem (MSP). 

If $n\geq n'$, let $P_i$ be a sub-problem generated by $\mathcal{A}$ with input integers of size $n_i\geq n'$, and such that all its ancestors sub-problems are computed, according to $\mathcal{A}$ using algorithms from the Toom-Cook class. We say that:
\begin{itemize}
    \item $P_i$ is a Type 1 $n'$-MSP of $\mathcal{A}$ if, according to $f_{\mathcal{A}}(P)$, it is computed using an algorithm from the standard class. If the entire problem is solved using an algorithm from the standard class, we say that the entire problem is the unique (improper) Type 1 $n'$-MSP generated by $\mathcal{A}$.
    \item $P_i$ is a Type 2 $n'$-MSP of $\mathcal{A}$ if, according to $f_{\mathcal{A}}(P)$, it is computed by  generating $2k-1$ sub-problems  using to the recursive Toom-$k$ scheme (where $k$ is determined by the instruction function $f_{\mathcal{A}}(P)$), and if the input size of all generated sub-problems is strictly smaller than $n'$. If the entire problem uses a recursive algorithm from the Toom-Cook class to generate sub-problems with input size smaller than $n'$, we say that the entire problem is the unique (improper) Type 2 $n'$-MSP generated by $\mathcal{A}$.
\end{itemize}
\end{definition}
In the following we denote as $\nu_1(n')$ (resp., $\nu_2(n')$) the number of Type 1 (resp., Type 2) $n'$-MSPs generated by $\mathcal{A}$, and $\nu(n')= \nu_1(n')+\nu_2(n')$.
For any given value $n'\in \mathbb{N}^+$, let $P_i$ denote the $i$-th $n'$-MSP generated by $\mathcal{A}$ and let $G^{\mathcal{A}}_{P_i}$ denote the corresponding sub-CDAG of $G^{\mathcal{A}}$. We denote as $A_i$ and $B_i$ (resp., $C_i$) the input  integers (resp., the output product) of $P_i$.\\

\noindent\textbf{Properties of MSPs and their corresponding sub-CDAGs:}  By Definition~\ref{def:maxsubp}, we have that for each pair of distinct $n'$-MSPs $P_1$ and $P_2$, $P_2$ is not recursively generated by $\mathcal{A}$ in order to compute $P_1$ or vice versa. Hence, by Lemma~\ref{lem:distinct}, the sub-CDAGs of $G^\mathcal{A}$ that correspond each to one of the $n'$-MSPs generated by $\mathcal{A}$ are vertex-disjoint.
In general, different $n'$-MSPs may have different input sizes.
For a given $n'$-MSP $P_i$ which multiplies input integers of size\ $n_i$, we denote as $\subpInput{}_i$ (resp., $\subpOutput{}_i$) the set of input (resp., output) vertices of the sub-CDAG $G^{\mathcal{A}}_{P_i}$ which correspond each to one of the $\lceil n_i/2\rceil$ least significant digits of the input integers $A_i$ and $B_i$ (resp., to one of the $\lceil n_i/2\rceil$ digits of the output product $C_i$ from the $(\lceil n_i/2\rceil+1)$-least significant to the $n_i$-least significant).  Finally, we define $\subpInput{} = \cup_{i=1}^{\nu(n')} \subpInput{}_i$ and $\subpOutput{} = \cup_{i=1}^{\nu(n')} \subpOutput{}_i$.


For each Type 1 $n'$-MSP $P_i$ generated by $\mathcal{A}$, with input integers of  size $n_i\geq n'$, we denote as $T_i$ the set of variables whose value correspond to the $\lceil n_i^2/4 \rceil$ elementary products  $A_i[j]B_i[k]$ for $j,k = 0,1,\ldots,\left(n_i-1\right)/2$ and $n_i/4\leq j+k\leq 3n_i/4$. By definition, these elementary products are computed by any integer multiplication algorithm from the standard class towards computing the output digits from the $\left(\halfsubpsize{}+1\right)$-least significant up to the $\subpsize$-th least significant one (i.e., the values corresponding to vertices in $\subpOutput{}_i$).
Further, we denote as $\mathcal{T}_i$ the set of vertices corresponding to the variables in $T_i$, and we define $\mathcal{T}= \cup_{i=1}^{\nu_1(n')} \mathcal{T}_i$.  Clearly, $|\mathcal{T}_i|\geq n_i^2/4\geq (n')^2/4$, and $|\mathcal{T}|\geq \sum_{i=1}^{\nu_i(n')}n_i^2/4\geq \nu_1(n')(n')^2/4$.

In order to obtain our \io{} lower bound for algorithms in $\hmm{}$, we characterize properties regarding the minimum dominator size of an arbitrary subset of $\subpInput{}$, $\subpOutput{}$, and $\mathcal{T}$. 
The proofs for the lemmas in this section are presented in Appendix~\ref{app:proof lemma}. 


\begin{lemma}\label{lem:domtype3}
For any Type 1 $n'$-MSPs generated by $\mathcal{A}$ consider $T'_i\subseteq T_i$.
Let $\subpInput{}_{A_i}\subseteq \subpInput{}_i$ (resp., $\subpInput{}_{B_i}\subseteq \subpInput{}_i$) denote a subset of the vertices corresponding to entries of $A_i$ (resp., $B_i$) which are multiplied in at least one of the elementary products in $T'_i$. Then any dominator $D$ of the vertices corresponding to $T'_i$ with respect to the the vertices in $\subpInput{}_i$ is such that
$$
    |D|\geq \max \{ |\subpInput{}_{A_i}|,|\subpInput{}_{B_i}|\}. 
$$
\end{lemma}

\begin{lemma}\label{lem:newconnect}
Let $G^\mathcal{A}$ be the CDAG corresponding to an algorithm $\mathcal{A}\in\hmm{}$ which generates $\nu(n')>0$ $n'$-MSPs. Let $P_i$ be one of such $n'$-MSPs.
For any  $Y\subseteq \subpInput{}_i$, any dominator set $D$ of $Y$ satisfies $|D|\geq |Y|$.
\end{lemma}

By composing the result of Lemma~\ref{lem:newconnect} and an analysis of the \emph{Partial Grigoriev's flow} of the integer multiplication function~\cite{BilardiS19} we obtain the following result:

\begin{lemma}\label{lem:domtype1}
Let $G^\mathcal{A}$ be the CDAG corresponding to an algorithm $\mathcal{A}\in\hmm{}$ which generates $\nu(n')$  $n'$-MSPs. Given any subset $Z\subseteq \subpOutput{}$ in $G^\mathcal{A}$ with $|Z|\leq n'/2$, any dominator set $\dom$ of $Z$ satisfies $|\dom|\geq|Z|/2$.
\end{lemma}
 Lemma~\ref{lem:domtype1} extends a result originally presented in~\cite{BilardiS19} for uniform, non-stationary  Toom-Cook algorithms, to the more general family of non-uniform, non-stationary hybrid algorithms considered in this work.

\section{\io{} lower bounds}\label{sec:mmlwb}
\begin{theorem} \label{thm:genmatmul}
Let $\mathcal{A}\in \hmm{}$ be an algorithm to multiply two integers   $A,B$ represented as $n$-digit base-$s$ numbers.
If run on a sequential machine equipped with a  cache of size $M$ and such that up to $B\leq M$  memory words stored in consecutive memory locations can be moved from cache to slow memory  and vice versa using a single memory operation, $\mathcal{A}$'s \io{} complexity satisfies:
\begin{equation}\label{eq:hmmmain}
	IO_{\mathcal{A}}\left(n,M,B\right) \geq 
	\max \{2n,|\mathcal{T}|/(4M),\nu(8M) M\}B^{-1},
\end{equation}
 where $4|\mathcal{T}|$ is the total number of internal elementary products computed by the Type 1 $8M$-MSPs  generated  by $\mathcal{A}$ and $\nu(8M)$ denotes the total number of $8M$-MSPs generated by $\mathcal{A}$.

If run on $\nproc{}$ processors each equipped with a local memory of size $M < n$ and where for each \io{} operation it is possible to move up to $\msgsize{}\leq M$ words, $\mathcal{A}$'s \io{} complexity satisfies:
\begin{equation}\label{eq:hmmpar}
	IO_{\mathcal{A}}\left(n,M,\nproc{},\msgsize{}\right) \geq 
	\max \{|\mathcal{T}|/(4M),\nu(8M) M\}\left(\nproc{}\msgsize{}\right)^{-1}.
\end{equation}
\end{theorem}
\begin{proof} 
We prove the result in~\eqref{eq:hmmmain} (resp.,~\eqref{eq:hmmpar}) for the case $B=1$ (resp., $\msgsize{}=1$). The result then trivially generalizes for a generic $B$ (resp., $\msgsize{}$). We first prove the result for the sequential case in~\eqref{eq:hmmmain}. The bound for the parallel case in~\eqref{eq:hmmpar} will be obtained as a simple generalization. 

The fact that $IO_{\alg}(n,M,1)\geq 2n$ follows from the fact that as in our model the input factor integers $A$ and $B$ are initially stored in slow memory, it will be necessary to move the entire input to the cache at least once using at least $2n$ \io{} operations.  If $\alg{}$ does not generate any $8M$-MSPs the statement in~\eqref{eq:hmmmain} is trivially verified. In the following, we assume $\nu(8M)\geq 1$. 

Let $G^{\mathcal{A}}$ denote the CDAG associated with algorithm $\mathcal{A}$ according to the construction in Section~\ref{sec:CDAG}. 
By Definition~\ref{def:maxsubp}, the  MSPs generated by $\alg{}$ have input (resp., output) integer factors whose base-$s$ expansion has at least $\subpsize{}$ digits. Recall that we denote as $\subpOutput{}$ the set of vertices which correspond to the digits of outputs of the $\nu(8M)$  $8M$-MSPs, starting from the $\left(\halfsubpsize{}+1\right)$-least significant up to the $\subpsize$-th least significant one.  We have $|\subpOutput{}|\geq \halfsubpsize{}\nu(8M)$. 
For each Type 1 $8M$-MSP $P_i$ generated by $\mathcal{A}$, with input integers of  size $n_i\geq 8M$, we denote as $T_i$ the set of variables whose value correspond to the at least $\lceil n_i^2/4\rceil$ elementary products  $A_i[j]B_i[k]$ for $j,k = 0,1,\ldots,\left(n_i-1\right)/2$ and $n_i/4\leq j+k\leq 3n_i/4$. By definition, these elementary products are computed by any integer multiplication algorithm from the standard class towards computing the output digits from the $\left(\halfsubpsize{}+1\right)$-least significant up to the $\subpsize$-th least significant one (i.e., the values corresponding to vertices in $\subpOutput{}$). Further, we denote as $\mathcal{T}_i$ the set of vertices corresponding to the variables in $T_i$, and we define $\mathcal{T}= \cup_{i=1}^{\nu_1(n')} \mathcal{T}_i$.  Clearly, $|\mathcal{T}_i|\geq n_i^2/4$, and $4|\mathcal{T}|= \sum_{i=1}^{\nu_1(8M)} n_i^2$.


Let $\mathcal{C}$ be any computation schedule for the sequential execution of $\mathcal{A}$ using a cache of size $M$. We partition $\mathcal{C}$ into non-overlapping segments  $\mathcal{C}_1,\mathcal{C}_2,\ldots$ such that during each $\mathcal{C}_j$ either \textbf{(a)} exactly $4M^2$ distinct values corresponding to vertices in $\mathcal{T}$, denoted as $\mathcal{T}^{(j)}$, are computed for the \emph{first time},
or \textbf{(b)} $4M$ distinct values corresponding to vertices in $\subpOutput{}$ (denoted as $\subpOutput{}_j$) are evaluated for the \emph{first time}. Clearly there are at least $\max \{|\mathcal{T}|/{4M^2},\nu(8M) \}$ such segments. 

Consider the set $\dom{}_j$ of vertices of $G_{\mathcal{A}}$ corresponding to the at most $M$ values stored in the cache at the beginning of $\mathcal{C}_j$ and to the at most $g_j$ values loaded into the cache form the slow memory (resp., written into the slow memory from the cache) during  $\mathcal{C}_j$ by means of a \texttt{read} (resp., \texttt{write}) \io{} operation. Clearly, $|\dom_j|\leq M +g_j$.
Below we show that the number $g_j$ of  \io{} operations  executed during each segment $\mathcal{C}_j$ satisfies $g_j\geq  M$, from which, as the segments are defined as non-overlapping, the theorem follows. \\

\noindent\textbf{Case (a):} For each Type 1 $8M$-MSP $P_i$  let $\mathcal{T}^{(j)}_i = \mathcal{T}^{(j)}\cap \mathcal{T}_i$. As the $\nu_1$ sub-CDAGs corresponding each to one of the Type 1 $8M$-MSPs are vertex-disjoint, so are the sets $\mathcal{T}_i$. Hence, the $\mathcal{T}^{(j)}_i$'s constitute a partition of $\mathcal{T}^{(j)}$. 
    Let $A_i$ and $B_i$ (resp., $C_i$) denote the input integers (resp., output product) of $P_i$, where $A_i$ and $B_i$ can be expressed in base-$s$ using up to $n_i$ digits.
    
    For $k = n_i/4,n_i/4+1,\ldots,3n_i/4-1$, we say that $C_i[k]$ is ``\emph{active during} $\mathcal{C}_j$'' if \emph{any} of the at least $n_i/4$ elementary products $A_i[r]B_i[s]$, for $0\leq r,s< n_i/2$ and $r+s=k$, correspond to any of the vertices in $\mathcal{T}^{(j)}_i$. Further we say that a $A_i[r]$ (resp., $B_i[r]$) is ``\emph{accessed during} $\mathcal{C}_j$'' if \emph{any} of the elementary multiplications $A_i[r]B_i[l]$ (resp., $A_i[l]B_i[r]$), for $l= 0,1,\ldots,n_i-1$, correspond to any of the vertices in $\mathcal{T}^{(j)}_i$. In the following we denote as $\alpha_i$ (resp., $\beta_i$) as the number of digits of $A_i$ (resp., $B_i$) which are accessed during $\mathcal{C}_j$, and let $\gamma_i$ denote the number of digits of $C_i$ which are active during $\mathcal{C}_j$. 
    By definition, $\max\{\alpha_i,\beta_i\}\geq \lceil\sqrt{|\mathcal{T}^{(j)}_i|}\rceil$ and $\gamma_i\geq \lceil |\mathcal{T}^{(j)}_i|/\max\{\alpha_i,\beta_i\}\rceil$.
    

Assume that for at least one of the $\nu_1(8M)$ Type 1 $8M$-MSPs, henceforth denoted as $P_i$, at least $2M$ among the $n_i/2$ least significant digits of $A_i$ or $B_i$ are accessed during $\mathcal{C}_j$. We denote as $Y$ the set of vertices of $G^{\mathcal{A}}$ corresponding to such values. Clearly $Y\subseteq\mathcal{Y}_i\subseteq \mathcal{Y}$. In order for the at least $2M$ values corresponding to vertices in $Y$ to be either stored in memory at the beginning of $\mathcal{C}_j$, read form slow memory to the cache during $\mathcal{C}_j$, or computed from other values during $\mathcal{C}_j$ there must be no path connecting any input vertex of $G^{\mathcal{A}}$ to any vertex in $Y$ which does not have at least one vertex in $\dom{}_j$. That is,  $\dom{}_j$ must be a \emph{dominator set} of $Y$, and,  by Lemma~\ref{lem:domtype3},
  $|D_j|\geq |Y|\geq 2M$.  
  Hence, we can conclude that if any of the active values $C[k]$ is computed entirely during $\mathcal{C}_j$ we have $g_j\geq M$.

Assume instead that for all Type 1 $8M$-MSPs $P_i$ generated by $\mathcal{A}$ strictly less than  $2M$ among the $n_i/2$ least significant digits of $A_i$ or $B_i$ are accessed during $\mathcal{C}_j$. That is $\max_{i\in\{1,2,\ldots,\nu_1(8M)\}}\{\alpha_i,\beta_i\}<2M$. Let $C_i[k]$ be active during $\mathcal{C}_j$. In order to compute $C[k]$ \emph{entirely} during $\mathcal{C}_j$ (i.e., without using partial accumulation of the summation $\sum_{r,s\geq 0|r+s=k} A_i[r]B[s]$), it will be necessary to evaluate all the at least $n_i/4$ elementary products which are added in the summation during $\mathcal{C}_j$ itself. Thus, at least $n_i/4\geq 2M$ values corresponding to digits from $A_i$ and $B_i$ must be accessed during $\mathcal{C}_j$. As, by assumption, for all $i\in\{1,2,\ldots,\nu_1(8M)\}$ we have $\alpha_i,\beta_i<2M$,  none of the values $C_i[k]$ which are active during $\mathcal{C}_j$ are computed entirely during $\mathcal{C}_j$ itself. As $\mathcal{C}$ is a parsimonious computation, for any value $C_i[k]$ which is active during $\mathcal{C}_j$ the values of the elementary products which are computed as part of the computation of $C_i[k]$ must be added to a partial accumulator corresponding to the summation tree of $C_i[k]$. Such an accumulator may have either already been previously initiated and be maintained in the memory (either cache or slow), or it may necessary to maintain it in memory (either cache or slow) at the end of $\mathcal{C}_j$. Thus, all the $\sum_{i=1}^{\nu_1}\gamma_i$ partial accumulators corresponding to active values $C_i[k]$ must be  included in the cache at the beginning of  $\mathcal{C}_j$, read form slow memory to the cache by means of a \texttt{read} \io{} operation, maintained in the cache at the end of $\mathcal{C}_j$, or written from the cache to the slow memory by means of a \texttt{write} \io{} operation. By the previous considerations we have:
  \begin{equation*}
      \sum_{i=1}^{\nu_1(M)}\gamma_i \geq \sum_{i=1}^{\nu_1(M)}\bigl\lceil\frac{|\mathcal{T}^{(j)}_i|}{\max\{\alpha_i,\beta_i\}}\bigr\rceil
      > \sum_{i=1}^{\nu_1(M)}\frac{|\mathcal{T}^{(j)}_i|}{2M}
      =\frac{|\mathcal{T}^{(j)}|}{2M}
      =2M.
  \end{equation*}
  Since at most $M$ of such accumulators may be in the cache at the beginning of (resp., remain in the cache at the of) $\mathcal{C}_j$, we can conclude that at least $M$ \io{} operations are executed during $\mathcal{C}_j$ (i.e., $g_j\geq M$).\\
   
\noindent\textbf{Case (b):} In order for the $4M$ values from $\subpOutput{}_j$ to be computed during $\mathcal{C}_j$ there must be no path connecting any vertex in $\subpOutput{}_j$ to any input vertex of $G_{\mathcal{A}}$ which does not have at least one vertex in $\dom_j$, that is  $\dom_j$ has to be a \emph{dominator set} of $\subpOutput{}_j$.
	 From Lemma~\ref{lem:domtype1}, any dominator set $D$ of any subset $Z\subseteq \subpOutput{}$ with $|Z|\leq 4M$ satisfies $|D|\geq |Z|/2$, whence $|D_j|=M+g_i\geq |\dom_i|\geq |\subpOutput{}_j|/2 =2M$, which implies $g_j\geq M$ as stated above.\\

	 This concludes the proof for the sequential case in~\eqref{eq:hmmmain}.

	The proof for the memory-dependent bound for the parallel model in ~\eqref{eq:hmmpar}, follows from the observation that at least one of the $\nproc{}$ processors, denoted as $\nproc{}^*$, must compute at least $|\mathcal{T}|/\nproc{}$ values corresponding to vertices in $\mathcal{T}$ or $|\subpOutput{}|/\nproc{}$ values corresponding to vertices in $\subpOutput{}$ (or both). The bound follows by applying the same argument discussed for the sequential case to the computation executed by $\nproc^*$. 
\end{proof}
 Note that if $\mathcal{A}$ is an uniform, non-stationary algorithm such that the product is entirely computed using algorithms from the Toom-Cook class (i.e.,  at each level of recursion the same version of Toom-$k$ -with the same value $k$- is used for all the sub-problems generated in all recursion branches, and at different levels of recursion different versions of the Toom-$k$ may be used), then $\mathcal{A}$ only generates Type 2 MSPs, and the bounds of Theorem~\ref{thm:genmatmul} correspond asymptotically to the results in~\cite{BilardiS19}.
 
 Instead, if $\mathcal{A}$ is such that the integer product is entirely computed using algorithms form the standard class, Theorem~\ref{thm:genmatmul} yields the following bound.
 
 \begin{corollary}[\io{} complexity of standard integer multiplication algorithms]\label{cor:standardio}
Let $\mathcal{A}$ be any standard integer multiplication algorithm which computes $\BOme{n^2}$ digit operations to multiply two integers $A,B$ represented as $n$-digit base-$s$ numbers on a sequential machine with  cache of size $M$ and such that up to $B\leq M$  memory words stored in consecutive memory locations can be moved from cache to slow memory  and vice versa using a single memory operation. The \io{} complexity of $\mathcal{A}$ satisfies:
\begin{equation}\label{eq:corstanda}
	IO_{\mathcal{A}}\left(n,M,B\right) \geq 
	\BOme{n^2/(BM)}.
\end{equation}
If run on $\nproc{}$ processors each equipped with a local memory of size $M < n$ and where for each \io{} operation it is possible to move up to $\msgsize{}\leq M$ words, $\mathcal{A}$'s \io{} complexity satisfies:
\begin{equation}\label{eq:corpar}
	IO_{\mathcal{A}}\left(n,M,\msgsize{},\nproc{}\right) \geq \BOme{n^2/(M\nproc{}\msgsize{})}.
\end{equation} 
 \end{corollary}

    
For the sub-class of uniform, non stationary algorithms $\uhmm{}$, given the values of $n$, $M$, $\nz{}$ and $k$, it is possible to compute a closed form expression for the values of $\nu_1,\nu_2$ and $|\mathcal{T}|$.
By applying  Theorem~\ref{thm:genmatmul} we have:
\begin{theorem}\label{thm:corgenmatmul}
Let $\mathcal{A}\in \uhmm{}$ be an algorithm to multiply two integers $A,B$ represented as $n$-digit base-$s$ numbers. If run on a sequential machine with  cache of size $M$ and such that up to $B\leq M$  memory words stored in consecutive memory locations can be moved from cache to slow memory  and vice versa using a single memory operation, $\mathcal{A}$'s \io{} complexity satisfies:
\begin{equation*}
	IO_{\mathcal{A}}\left(n,M,B\right) \geq \BOme{\left(\frac{n}{\max\{\nz{},M\}}\right)^{\log_k \left(2k-1\right)} \left(\max\bigl\{1, \frac{\nz{}}{M}\bigr\}\right)^2 \frac{M}{B}}
\end{equation*}
If run on $\nproc{}$ processors each equipped with a local memory of size $M < n^2$ and where for each \io{} operation it is possible to move up to $\msgsize{}\leq M$ memory words, $\mathcal{A}$'s \io{} complexity satisfies:
\begin{equation*}
	IO_{\mathcal{A}}\left(n,M,\nproc{},\msgsize{}\right) \geq \BOme{\left(\frac{n}{\max\{\nz{},M\}}\right)^{\log_k \left(2k-1\right)} \left(\max\bigl\{1, \frac{\nz{}}{M}\bigr\}\right)^2 \frac{M}{\nproc{}\msgsize{}}}
\end{equation*}
\end{theorem}

\begin{proof}
As until the size of the sub-problems is higher than $\max\{n_0,8M\}$ $\mathcal{A}$  behaves as an uniform Toom-$k$ algorithm, we have that the inputs of the $(2k-1)^i$ sub-problems generated in the $i$-th level of recursion have size at least $n/k^i$.

Let $\ell$ be the  minimum level of recursion levels $\ell$, \emph{necessary} for the size of the input of the sub-problems to fall below $\max\{n_0,8M\}$. We have $\ell= \lceil \log_k \frac{n}{\max\{n_0,8M\}}\rceil$.
Then, by the definition of MSP, either $\mathcal{A}$ generates 
$\nu_2\geq \BOme{\left(n/M\right)^{\log_k (2k-1)}}$
Type 2 MSPs, or $\mathcal{A}$ generates
$\nu_1\geq \BOme{\left(n/\nz{}\right)^{\log_k (2k-1)}}$
vertex disjoint Type 1 MSP, each with input size \emph{at least} $\nz$. Thus 
    $|\mathcal{T}|\geq \BOme{\left(n/\nz{}\right)^{\log_k (2k-1)}\nz^2}$.
\end{proof} 
The \io{} complexity lower bounds for the parallel setting presented in~\eqref{eq:hmmpar}, and its special cases  in Corollary~\ref{cor:standardio} and Theorem~\ref{thm:corgenmatmul}, are memory-dependent as we assume that each processor is equipped with a cache which can hold up to $M$ memory words. For $\msgsize{}=1$ these results provide  lower bounds to the \emph{bandwidth cost} (i.e., number of memory words exchanged by at least one processor) of the respective class of algorithms.  

\subsection{Memory-independent \io{} lower bounds in the parallel model\\ assuming balanced input distribution}\label{sec:mind}
We now present memory-independent \io{} lower bounds for parallel algorithms in $\hmm{}$. For these results, we do not assume any constraint on the size of the local memories of the $P$ used processors (which can be assumed to be unbounded). Rather, we assume a \emph{balanced distribution} of the input integers among the $P$ processors at the beginning of the computation. We do not, however, assume any further balance for the overall computational load of the algorithm, nor for the distribution of the output at the end of the computation. 

\begin{theorem}\label{thm:memindbound}
Let $\mathcal{A}\in \hmm{}$ be an algorithm to multiply $A,B$ represented as $n$-digit base-$s$ numbers using $\nproc{}$ processors each equipped unbounded local memory. Assume further that at the beginning of $\mathcal{A}$ no processor has more than $\alpha n/\nproc{}$, where $\alpha \in [1,P]$,  digits of each of the input integers stored in its local memory, and that for each \io{} operation it is possible to move up to $\msgsize{}$ memory words. $\mathcal{A}$'s \io{} complexity satisfies:

\begin{equation*}
	IO_{\mathcal{A}}\left(n,\nproc{},\msgsize{}\right) \geq \BOme{\frac{\max_{n'\in \{\alpha n/P,\ldots,n\}}n'\min\{\frac{1}{4}, \sqrt{\frac{\nu_i(n')}{8P}}+\frac{\nu(n')}{2\nproc{}}\}-\frac{\alpha n}{P}}{\msgsize{}}}
\end{equation*}
\end{theorem}
 The main idea for the proof of Theorem~\ref{thm:memindbound} draws inspiration from the reasoning used in Theorem~\ref{thm:genmatmul}: given any $\mathcal{A}\in \hmm$, for $n'\in \{1,2,\ldots,n\}$ at least one of the $P$ processors, denoted as $P^*$, must compute either at least $(n')^2\nu_1(n')/(4P)$ of the elementary products of the $\nu_1(n')$ Type 1 $n'$-MSPs generated by $\mathcal{A}$, or at least $n'\nu (n')/P$ output values of the $\nu(n')$ $n'$-MSPs generated by $\mathcal{A}$. As discussed in the proof of Theorem~\ref{thm:genmatmul} any dominator set $D$ for the vertices corresponding to these values must be such that $D\geq n'\sqrt{\nu_1(n')}/(2\sqrt{P})$(resp., $|D|\geq n'\nu (n')/P$), and hence there are at least $D$ vertex-disjoint paths connecting the vertices corresponding to the values computed by $P^*$ to input vertices of $G^{\mathcal{A}}$. At most  $\alpha n/\nproc{}$ of these paths may be dominated by vertices corresponding to the at most $\alpha n/\nproc{}$ input values initially stored in $P^*$ or which may be computed by $P^*$ without communicating with other processors. For each remaining path it will then necessary to receive at least a value from one of the other processors. The detailed proof of Theorem~\ref{thm:memindbound} is presented in  Appendix~\ref{app:proofmind}.

For the sub-class of uniform, non stationary algorithms $\uhmm{}$, given the values of $n$, $M$, $\nz{}$ and $k$, it is possible to compute a closed form expression for the values of $\nu_1(n')$ and $\nu_2(n')$ for $n'= n/P^{1/\log_k (2k-1)}$ and for $n'=\nz{}$. In particular, if $n_0< n/P^{1/\log_k (2k-1)}$ we have $\nu(n/P^{1/\log_k (2k-1)}) = \Theta\left(P\right)$. Otherwise, we have $\nu_1(\nz{})= \Theta\left(\left( n/\nz{} \right)^{\frac{\log_k (2k-1)}{2}}\right)$.
Then, by applying  Theorem~\ref{thm:memindbound} we have:

\begin{theorem}\label{thm:specmemindbound}
Let $\mathcal{A}\in \uhmm{}$ be an algorithm to multiply $A,B$ represented as $n$-digit base-$s$ numbers using $\nproc{}$ processors each equipped unbounded local memory. Assume further that at the beginning of $\mathcal{A}$ no processor has more than $\alpha n/\nproc{}$ (where $\alpha$ is a constant with respect to $n/\nproc{}$) digits of each of the input integers stored in its local memory, and that for each \io{} operation it is possible to move up to $\msgsize{}$ memory words. $\mathcal{A}$'s \io{} complexity satisfies:
\begin{equation*}
	IO_{\mathcal{A}}\left(n,\nproc{},\msgsize{}\right) \geq \begin{cases} \BOme{\frac{n}{\msgsize{}P^{1/\log_k (2k-1)}}}, &if\ \frac{n}{P^{1/\log_k (2k-1)}}\geq n_0\\
	\BOme{\left(\frac{n}{n_0}\right)^{\frac{\log_k (2k-1)}{2}} \frac{n_0}{\msgsize{}\sqrt{\nproc{}}}}, &if\ n_0> \frac{n}{P^{1/\log_k (2k-1)}}.
	\end{cases} 
\end{equation*}
\end{theorem}

 For the case $nP^{-1/\log_k (2k-1)}\geq n_0$, the result of Theorem~\ref{thm:memindbound} yields a memory-independent \io{} lower bound for the class of parallel ``\emph{uniform}'' Toom-Cook algorithms studied in~\cite{BilardiS19}.

As, by definition of Type 1 MSP, any algorithm from the standard class used to multiply $n$-digits integers generates a single Type 1 MSP $n$-MSP, form Theorem~\ref{thm:memindbound} we obtain a memory-independent lower bound on the \io{} complexity of any parallel implementation of an integer multiplication algorithm from the standard class:

 \begin{corollary}\label{cor:mindsta}
Let $\mathcal{A}$ be any standard integer multiplication algorithm to multiply two integers $A,B$ represented as $n$-digit base-$s$ numbers using $\nproc{}$ processors each equipped with an unbounded local memory. Assume further that at the beginning of $\mathcal{A}$ no processor has more than $\alpha n/\nproc{}$ (where $\alpha$ is a constant with respect to $n/\nproc{}$) digits of each of the input integers stored in its local memory, and that for each \io{} operation it is possible to move up to $\msgsize{}$ memory words. $\mathcal{A}$'s \io{} complexity satisfies:
\begin{equation*}
	IO_{\mathcal{A}}\left(n,\nproc{},\msgsize{}\right) \geq \BOme{\frac{n}{\msgsize{}\sqrt{\nproc{}}}}.
\end{equation*} 
 \end{corollary}
 
 For $\msgsize{}=1$, the results in this section yield  lower bounds to the \emph{bandwidth cost} (i.e., number of memory words exchanged by at least one processor) of the respective class of algorithms.

\subsection{Memory-independent \io{} lower bounds in the parallel model\\ assuming balanced computation}\label{sec:mindcomp}
We now present alternative memory-independent \io{} lower bounds for parallel algorithms in $\hmm{}$. For these results, we assume that the processors share the computational load of the algorithm in a \emph{balanced} way. However, we do not assume any bound on the available local memory (which can be assumed to be unbounded). Differently from the previous case, we do not assume that the input is initially partitioned in a balanced fashion among the processor, but only that the input values are not replicated among multiple processors at the beginning of the computation (i.e., we assume that at the beginning of the computation each digit of the input integers is stored in the local memory of a single processor). Such criterion is of natural interest as a balanced distribution of the computational load among the processors is necessary in order to achieve optimal speedup. 

We say that an algorithm $\mathcal{A}$ from the standard class is $\beta$-computationally balanced, for $\beta \in (0,1]$, if, when computing the product of two $n$-digit integers, each processor computes at least a fraction $\beta\frac{n^2}{\nproc{}}$ of the $n^2$ \emph{elementary products} among digits of the input to be computed. Clearly, any algorithm for which at least one of the processors computes a fraction $\beta\frac{n^2}{\nproc{}}$ of the $n^2$ elementary products will have parallel execution time at least $\BOme{\frac{n^2(\nproc{}-\beta)}{\nproc{}(\nproc{}-1)}}$. Towards achieving optimal speedup it therefore appears desirable to have $\beta$ close to one. 

\begin{theorem}\label{thm:balancompsta}
Let $\mathcal{A}$ be any $\beta$-computationally balanced standard integer multiplication algorithm to multiply two integers $A,B$ represented as $n$-digit base-$s$ numbers using $\nproc{}$ processors each equipped with an unbounded local memory. Assume further that at the beginning of $\mathcal{A}$  each digit of the input integers is stored in the local memory of a single processor. $\mathcal{A}$'s \io{} complexity satisfies:
\begin{equation*}
	IO_{\mathcal{A}}\left(n,\nproc{},\msgsize{}\right) \geq \BOme{\sqrt{\beta}\frac{n}{\sqrt{\nproc{}}\msgsize{}} - \frac{2n}{\nproc{}\msgsize{}}}
\end{equation*}
Where $\msgsize{}$ denotes the number of memory words which can be transmitted in a single \io{} operation.
 \end{theorem}
\begin{proof}
We prove the statement for $\msgsize{}=1$. The general case follows a simple generalization.
By the assumption of non-replication of the input, at least one of the $\nproc{}$ processor, henceforth denoted as $\nproc{}^*$, holds no more than $2n/\nproc{}$ of the $2n$ digits of the input. As, $\mathcal{A}$ is an $\beta$-computationally balanced algorithm, by definition, $\nproc{}^*$ computes at least $\beta\frac{n^2}{\nproc{}}$ of the $n^2$ elementary products required for obtaining the product $C=A\times B$. Let $G^{\mathcal{A}}$, be the CDAG corresponding to algorithm $\mathcal{A}$. The input vertices of $G^{\mathcal{A}}$ correspond each to one of the digits of the input integers. Among the vertices of $G^{\mathcal{A}}$, there are $n^2$ which correspond each to one of the elementary products computed by the algorithm. Let $T$ denote the set of vertices corresponding to the at least $\beta\frac{n^2}{\nproc{}}$ elementary product computed by $\nproc{}^*$. By Definition~\ref{def:maxsubp}, the entire integer multiplication problem being considered is the unique, improper, Type 1 $n$-MSP generated by $\mathcal{A}$. 
By the definition of the standard algorithm class, among the $n^2$ elementary products, each digit of $A$ is multiplied by all $n$ distinct digits of $B$ and vice versa. Hence, the number $a$ (resp., $b$) of digits of $A$ (resp., $B$) a which are multiplied in at least one of the $\beta\frac{n^2}{\nproc{}}$ elementary products computed by $\nproc{}^*$ must satisfy $\max\{a,b\}\geq \sqrt{|T|} = \sqrt{\beta\frac{n^2}{\nproc{}}}$.
Thus, by Lemma~\ref{lem:domtype1}, any dominator $D$ of $T$ must satisfy $|D|\geq \max\{a,b\}$.
As previously discussed, at most input values corresponding to the digits of the input integers $A$ and $B$ are stored in the local memory of $\nproc{}$ at the beginning of the computation of $\mathcal{A}$.    By Definition~\ref{def:dominator}, this implies that there are at least $ n\sqrt{\beta,\nproc{}}$ vertex-disjoint paths connecting vertices in $T$ to the input vertices of $G^{\mathcal{A}}$, and, in particular, at least $\sqrt{\beta,\nproc{}}-2n/\nproc{}$ such paths connect vertices of $T$ to input vertices corresponding to values not initially stored in the local memory of $\nproc{}^*$.  Let $L$ denote the set of vertices corresponding to values which are either received or sent by $\nproc{}^*$ during the computation of $\mathcal{A}$ by communicating with other processors. In order for $\nproc{}^*$ to compute the values corresponding to the vertices in $T$, there must be no path connecting inputs of $G^{\mathcal{A}}$ not initially stored in the local memory of $\nproc{}^*$ to vertices in $T$. Thus, by the previous considerations, we  have
\begin{equation*}
    |L|\geq  \sqrt{\beta\frac{n^2}{\nproc{}}} - \frac{2n}{\nproc{}}.
\end{equation*}
\end{proof}
The result in Theorem~\ref{thm:balancompsta} provides insight into the trade off between the balance of the computational load and the \io{} complexity of standard integer multiplication algorithms in the parallel distributed memory-setting: the more balanced the computation (i.e., for $\beta$ close to 1) the \io{} requirement grows. Vice versa, if the computation is distributed in rather unbalanced way (i.e., for $\beta$ close to 0), the \io{} requirement is lower but the parallel computation time, bounded by $\BOme{\frac{n^2(\nproc{}-\beta)}{\nproc{}(\nproc{}-1)}}$, grows correspondingly. 

The result in Theorem~\ref{thm:balancompsta} can be straightforwardly extended to hybrid algorithms in the uniform-non stationary class $\uhmm{}$, for $n_0> n{P^{-1/\log_k (2k-1)}}$. Algorithms in this class use standard integer multiplication algorithms to compute the  $\BOme{\left(\frac{n}{n_0}\right)^{\frac{\log_k (2k-1)}{2}} n_0^2}$ elementary products of the input values of the $\BOme{\left(\frac{n}{n_0}\right)^{\frac{\log_k (2k-1)}{2}}}$ sub-problems with input size at least $n_0$ generated after the initial $\BOme{\log_k n/n_0}$ recursive steps using the Toom-$k$. We say that an algorithm $\mathbf{A}$ is $\beta$-balanced , for $\beta\in(0,1]$, if each processor computes at least a fraction $\BOme{\beta\left(\frac{n}{n_0}\right)^{\frac{\log_k (2k-1)}{2}} \frac{n_0^2}{\nproc{}}}$ of the \emph{elementary products}. 
Following the reasoning in Theorem~\ref{thm:memindbound} and Theorem~\ref{thm:balancompsta} we have:

\begin{corollary}\label{cor:specmemindboundbalanced2}
Let $\mathcal{A}\in \uhmm{}$, where $n_0> n{P^{-1/\log_k (2k-1)}}$, be a $\beta$-balanced algorithm to multiply $A,B$ represented as $n$-digit base-$s$ numbers using $\nproc{}$ processors each equipped unbounded local memory. Assume further that at the beginning of $\mathcal{A}$  each digit of the input integers is stored in the local memory of a single processor. $\mathcal{A}$'s \io{} complexity satisfies:
\begin{equation*}
	IO_{\mathcal{A}}\left(n,\nproc{},\msgsize{}\right) \geq 
	\BOme{\left(\frac{n}{n_0}\right)^{\frac{\log_k (2k-1)}{2}} \frac{\sqrt{\beta}n_0}{\msgsize{}\sqrt{\nproc{}}}}
\end{equation*}
Where $\msgsize{}$ denotes the number of memory words which can be transmitted in a single \io{} operation.
\end{corollary}
 
 For $\msgsize{}=1$, the results in this section yield  lower bounds to the \emph{bandwidth cost} (i.e., number of memory words exchanged by at least one processor) of the respective class of algorithms.
 
 Interestingly, for values of $\beta$ (resp., $\alpha$) which are constant with respect to $n,n_0$ and $\nproc{}$, the bound  in Theorem~\ref{thm:balancompsta} (resp., Corollary~
 \ref{cor:specmemindboundbalanced2}) corresponds asymptotically to that obtained for the input-balanced parallel computations of standard integer algorithms (resp., algorithms in $\uhmm{}$, where $n_0> \frac{n}{P^{1/\log_k (2k-1)}}$) in Corollary~\ref{cor:mindsta} (resp., Theorem~\ref{thm:specmemindbound}). This suggests that the two notion of \emph{balance} being considered have a similar effect on the \io{} complexity of integer multiplication algorithms.

\subsection{On the tightness of the bounds} 

As  \io{} efficient sequential implementations of the \emph{Grid method multiplication} execute at least $\BO{n^2/M}$ \io{} operations, we can conclude that such implementations are asymptotically optimal, and that the \io{} lower bound in~\eqref{eq:corstanda} for the sequential model is asymptotically tight.

In~\cite{BilardiS19}, Bilardi and De Stefani presented a sequential version of the Toom-Cook algorithm which is ``uniform, non-stationary'', henceforth referred to as $\mathcal{A}_{\mathcal{BD}}$. At each level of recursion $\mathcal{A}_{\mathcal{BD}}$ uses the same version of Toom-$k$ (with the same value $k$)for all the sub-problems generated in all recursion branches, while different versions of the Toom-$k$ scheme may be used in different levels of recursion. The algorithm proceeds according to such a scheme until one such recursive level for which all the generated sub-problems can be computed in the cache. According to our definition of MSP, $\mathcal{A}_{\mathcal{BD}}$ generates only Type-2 MSPs, and as shown in~\cite[Theorem 5.1 and Lemma 5.2]{BilardiS19}, the number of \io{} operations executed by it is within a $\BO{k_{max}^2}$ multiplicative of the \io{} lower bound in Theorem~\ref{thm:genmatmul}, provided that the size of the available cache memory is such that $M\geq \BOme{k^3\log_sk}$. 
Through simple, albeit tedious, modifications of algorithm $\mathcal{A}_{\mathcal{BD}}$, is it possible to obtain non-uniform, non-stationary algorithms for integer multiplication using the Toom-Cook scheme $\hmm{}$ whose \io{} complexity is is within a $\BO{k_{max}^2}$ multiplicative of the \io{} lower bound in Theorem~\ref{thm:genmatmul}. Let us refer to such modified algorithms as $\mathcal{A}^*_{\mathcal{BD}}$. $\mathcal{A}^*_{\mathcal{BD}}$ follows the same steps as $\mathcal{A}_{\mathcal{BD}}$ but, according to the instruction function $f_{\mathcal{A}^*_{\mathcal{BD}}}$, may use different versions of Toom-$k$ even for sub-problems with inputs of the same size generated at the same levels of recursion. $\mathcal{A}^*_{\mathcal{BD}}$ computes entirely in cache sub-problem whose input size fits in the available cache.  

An opportune composition of algorithm $\mathcal{A}^*_{\mathcal{BD}}$ (for different choices of the associated instruction function) with the \io{} optimal sequential Grid method multiplication algorithm lead to  sequential hybrid algorithms  in $\hmm{}$ (resp., $\uhmm{}$) whose \io{} cost asymptotically is within an $\BO{k_{max}^2}$ multiplicative factor of the \io{} complexity sequential lower bounds in Theorem~\ref{thm:genmatmul} and Theorem~\ref{thm:corgenmatmul}. For constant values of $k_{max}$ (with respect to $n,M$), such hybrid algorithms are \io{} optimal and the corresponding lower bounds are asymptotically tight. Interestingly, as all the mentioned algorithms from $\hmm{}$ and $\uhmm{}$ do not recompute any intermediate value, we can conclude that using recomputation may lead to at most a $\BO{k_{max}^2}$ multiplicative factor reduction of the \io{} cost of hybrid algorithms in $\hmm{}$ and $\uhmm{}$.

In~\cite{lds21}, De Stefani presented parallel versions of the standard recursive integer multiplication algorithm and of Karatsuba's algorithm for the distributed memory setting. Under mild conditions (i.e., $n\geq \nproc$ and $M\geq \log_2 \nproc$), these algorithms exhibit optimal computational speedup and their bandwidth cost (i.e., the number of memory words being either sent or received by at least one processor) matches asymptotically the corresponding lower, both memory-dependent and memory-independent bounds, discussed in this work when considering Karatsuba's algorithm as a special case of Toom-$2$. Further, the \emph{latency cost} of the mentioned algorithms is within a $\BO{\log_2^2 \nproc{}}$ multiplicative factor of the corresponding lower bounds in this paper. Interestingly, the algorithms in~\cite{lds21} employ a balanced distribution of the input, and a balanced distribution of the computational load. By opportunely combining these algorithms, it is is possible to obtain hybrid uniform, non-stationary integer multiplication algorithms combining standard and Karatsuba's algorithm. The bandwidth cost of such hybrid algorithms matches asymptotically the corresponding lower bounds discussed in this work (i.e., Theorem~\ref{thm:corgenmatmul}, Theorem~\ref{thm:specmemindbound} and Corollary~\ref{cor:specmemindboundbalanced2}), while their latency cost is within a $\BO{\log_2^2 \nproc{}}$ multiplicative factor of the proposed lower bounds.

Hence, we have that our lower bounds for the parallel setting (for the case Toom-$2$), both memory-dependent and memory-independent, are asymptotically tight for what pertains the bandwidth cost, and within a multiplicative $\BO{\log_2^2 \nproc{}}$ of the optimal latency.

\section{Conclusion}
This work presented a characterization of the \io{}
complexity of a general class of non uniform, non stationary hybrid integer multiplication algorithms combining fast Toom-Cook algorithms with standard algorithms. Our work further showcase the generality of the approach based on the analysis of \emph{Maximal Sub-Problems} introduced in~\cite{dest19} by showing how it can be used to also obtain memory -independent \io{} lower bounds in the parallel setting. While the presented lower bounds are (almost) asymptotically
tight for the sequential setting, an important open question is whether it is possible to a general parallel version of Toom-Cook algorithms whose \io{} cost (both bandwith and latency) asymptotically matches  the lower bounds presented in this work.


\clearpage
\bibliographystyle{plain}
\bibliography{bibliography}
 \cleardoublepage
 \appendix
 \section{Proofs of properties of MSP sub-CDAGs}\label{app:proof lemma}

\begin{proof}[Proof of Lemma~\ref{lem:domtype3}]
Without loss of generality, let us assume $ |\subpInput{}_{A_i}|\geq|\subpInput{}_{B_i}|$. The proof for the case $|\subpInput{}_{A_i}|<|\subpInput{}_{B_i}|$ follows an analogous argument.
Let $\dom{}$ be a dominator set for the set of vertices corresponding to $T_i'$ with respect to
$\subpInput{}_{A_i}$. Consider a possible assignment to the values of $B_i$ such that all the values corresponding to vertices in $\subpInput{}_{B_i}$ are assigned value 1. 
Under such assignment, for every variable $a\in A_i$ corresponding to a vertex in $\subpInput{}_{A_i}$ at least one of the elementary products in $T_i'$ assumes value $a$. The lemma follows combining statements (i) and (ii):\\ (i) There exists an
assignment of the input variables of $P_i$ corresponding to vertices in $\subpInput{}_i\setminus \subpInput{}_{A_i}$
such that the output variables in $T_i'$ assume at least
$s^{|\subpInput{}_{A_i}|}$ distinct values under all possible assignments of the variables corresponding to vertices in $\subpInput{}_{A_i}$.\\ (ii) Since all paths form $\subpInput{}_{A_i}$ to the vertices corresponding to the variables in $T_i'$ intercept $\dom$, the values of the elementary products in $T_i'$ are determined by the inputs in
$\subpInput{}_i\setminus \subpInput{}_{A_i}$, which are fixed, and by the values of the
vertices in $\dom{}$; hence the elementary products in $T_i'$ can take at most
$s^{|\dom|}$ distinct values. 
\end{proof}

In the following we denote as $\globalInput{}$ the set of input vertices of the CDAG $G^{\mathcal{A}}$. That is, $\mathcal{X}$ is the set of vertices which corresponds each to one of the digits in the base-$s$ expression of the input $n$-digit integers $A$ and $B$. 

\begin{proof}[Proof of Lemma~\ref{lem:newconnect}]
Let $\ell$ denote the number of recursive steps of algorithm $\mathcal{A}$ at which the MSP $P_i$ is generated. The proof is by induction on $\ell$.
In the base case $\ell=0$, that is $P_i$ is the only (improper) $n'$-MSP generated by $\mathcal{A}$. Hence, $\subpInput{}$ and $\globalInput{}$ coincide and the statement is trivially verified. 
We assume inductively that the statement holds for $\ell = j \geq 0$, and we show it holds also for $\ell =j+1$.
As $\ell>0$, we have that in the first level of recursion algorithm, $\mathcal{A}$ generates $2k-1$ sub-problems, $P^{(1)},P^{(2)},\ldots,P^{(2k-1)}$, where the value $k$ is given by the instruction function associated with $\mathcal{A}$.
Let us assume, without loss of generality, that $P_i$ is recursively generated by $\mathcal{A}$ while solving $P^{(1)}$, and let $G^{\mathcal{A}_{P_1}}$ be the sub-CDAG of $G^{\mathcal{A}}$ corresponding to $P_1$. By inductive hypothesis, and by the definition of dominator set,  there exists a subset $K$ of the input vertices of $G^{\mathcal{A}_{P_1}}$ with $|K|=|Y|$ such that there exists $|K|$ vertex-disjoint paths connecting vertices of $K$ to vertices in $Y$.
From the construction of  $G^{\mathcal{A}}$, that the input vertices of $G^{\mathcal{A}_{P_1}}$, and hence the vertices in $K$, are connected to the vertices in $\globalInput$ by means of an encoder sub-CDAG $Enc_{k,n}$ of $G^{\mathcal{A}}$, where $n$ is the  number of digits in the base-$s$ expression of the input integers $A$ and $B$.    

From Lemma~\ref{lem:conneconder}, we have that there exists a subset of the input vertices $X$, with $|X| = |K|$, such that there exist $|K|$ vertex-disjoint paths connecting the vertices in $X$ to the vertices in $K$. By composing the $|X| = |K|$ vertex-disjoint paths connecting vertices in $X$ to vertices $K$ with the $|K|=|Y|$ vertex-disjoint paths connecting vertices in $K$ to vertices in $Y$, we can conclude that there exist indeed $|X| = |Y|$ vertex-disjoint paths connecting vertices in $X$ to vertices in $Y$. Let $D$ be any dominator of $Y$. From the existence of $X$, we have $|D|\geq |X|=|Y|$.
\end{proof}

\begin{proof}[Proof of Lemma~\ref{lem:domtype1}]
If $\nu(n')=0$, the statement is vacuously true. In the following, we assume $\nu(n')\geq 1$.

Assume towards contradiction that there exists a dominator $D$ for $Z$ such that $|D|< |Z|/2$. Thus, any path connecting input an input vertex of  $G^\mathcal{A}$ to a vertex of $Z$ must include at least one vertex in $D$.  
Let $G^{\mathcal{A}_{P_1}},G^{\mathcal{A}_{P_2}}\ldots,G^{\mathcal{A}_{P_{\nu(n')}}}$ denote the sub-CDAGs of $G^{\mathcal{A}}$, each corresponding to one of the $\nu(n')$  $n'$-MSPs generated by $\mathcal{A}$. We define as $D'$  the subset of $D$ composed by vertices \emph{internal} to the sub-CDAGs $G^{\mathcal{A}_{P_i}}$, excluding their respective inputs
For $i=1,2,\ldots,\nu(n')$, let $Z_i$ (resp., $D_i$) denote the subset of $Z$ (resp., $D'$) composed by vertices in $G^{\mathcal{A}_{P_i}}$. As, from Lemma~\ref{lem:distinct}, the sub-CDAGs $G^{\mathcal{A}_{P_i}}$ are vertex-disjoint, $\{Z_1,Z_2,\ldots,Z_{\nu(n')}\}$ partition $Z$ (resp., $\{D'_1,D'_2,\ldots,D'_{\nu(n')}\}$ partition $D'$). Let:
\begin{equation*}
    i^* = \argmin_{i=1,2,\ldots,\nu(n')} |D'_i|/|Z_i|.
\end{equation*}
In the following, we denote as $n_{i^*}$ the input size of the $n'$-MSP $P_{i^*}$, that is, $n_{i^*}$ denotes the number of digits of the base-$s$ expansion of each of the input integers of $P_{i^*}$. 

Let $Y'\subseteq \subpInput{}_{i^*}$ be such that all paths connecting the vertices in $Y'$ to the output vertices in $Z_{i^*}$ include at least a vertex in $\dom{}_{i^*}$ (i.e., $Y'$ is the largest subset of $\subpInput{}_{i^*}$ with respect to whom $\dom{}_{i^*}$ is a dominator set for $Z_{i^*}$). 
From Lemma~\ref{cor:flowcor} we have:
	\begin{equation*}
		|\dom{}'_{i^*}|\geq |Y'||Z_{i^*}|/n_{i^*},
	\end{equation*}
which implies $|Y'|\leq n_{i^*}|D'_{i^*}|/|Z_{i^*}|$.

Let $Y =\subpInput{}_{i^*} \setminus Y'$. By the definition of $Y'$, the vertices in $Y$ are those among the ones in $\subpInput{}_{i^*}$ that are connected to vertices in $Z_{i^*}$ by directed paths with no vertex in $D_{i^*}$. We have:
\begin{equation*}
	|Y| \geq |\subpInput{}_{i^*}|- n_{i^*}|D'_{i^*}|/|Z_{i^*}| \geq n_{i^*}\left(1/2-|D'_{i^*}|/|Z_{i^*}|\right),
\end{equation*}
where the last passage follows from the fact that, by definition of $\subpInput{}_{i^*}$, we have $|\subpInput{}_{i^*}|\geq  n_{i^*}/2$.
According to a well-known property of the ratios of sums we have:
    \begin{equation*}
    	\frac{|D'_{i^*}|}{|Z_{i^*}|} = \min_{i\in\{1,2,\ldots,\nu_2(n')\}} \frac{|D'_i|}{|Z_i|} \leq \frac{\sum_{i=1}^{\nu_2(n')} |D'_i|}{\sum_{i=1}^{\nu_2(n')} |Z_i|} = \frac{|D'|}{|Z|},
    \end{equation*}
    and thus:
\begin{equation*}
	|Y| = n_{i^*}\left(1/2-|D'_{i^*}|/|Z_{i^*}|\right) \geq n_{i^*}\left(1/2-|D'|/|Z|\right),
\end{equation*}

By Lemma~\ref{lem:newconnect}, there exists a set $X\subseteq\globalInput{}$, where $\mathcal{X}$, denotes the set of input vertices of $G^{\mathcal{A}}$ with $|X|=|Y|$, such that vertex in $X$ are connected to vertices in $Y$ 
by $|X|=|Y|$ vertex-disjoint paths. As such paths are vertex-disjoint, at most $|D\setminus D'|$ of these paths may include at least a vertex in $D\setminus D'$. Hence, we have that there exist $\pi$ paths connecting vertices in $Z$ to vertices in $\globalInput{}$ where:
   \begin{align*}
		\pi & \geq   n_{i^*}\left(1/2-|D'|/|Z|\right)-|\dom\setminus\dom'|\\
		&\geq n_{i^*} \left(1/2 - |\dom|/|Z|\right),
	\end{align*}  
where the last passage follows as $P_{i^*}$ is a $n'$-MSP and, thus,  by Definition~\ref{def:maxsubp}, $n_{i^*}\geq n'\geq 2|Z|$.
For $|D|\leq|Z|/2-1$, we have that $\pi\geq 1$, which contradicts $D$ being a dominator set for $Z$.
\end{proof}

\section{Proof of Theorem~\ref{thm:memindbound}}\label{app:proofmind}

\begin{proof}[Proof of Theorem~\ref{thm:memindbound}]
We prove the result for the case $\msgsize{} = 1$. The bound then trivially generalizes for $\msgsize{}>1$.

Let $G^{\mathcal{A}} = \left(V,E \right)$ denote the CDAG associated with algorithm $\mathcal{A}$ according to the construction in Section~\ref{sec:CDAG}.

Recall that $\nu_{1}(n')$ (resp., $\nu_2(n')$) denotes the number of Type 1 (resp., Type 2) $n'$-MSP generated by $\mathcal{A}$, and that $\nu(n')=\nu_{1}(n')+\nu_{2}(n')$. By Definition~\ref{def:maxsubp}, each of these MSPs has input size greater than or equal to $n'$. 


For each $n'$-MSP $P_i$  we denote as $\subpInput{}_i$ (resp., $\subpOutput{}_i$) the set of input (resp., output) vertices of the sub-CDAG $G^{\mathcal{A}}_{P_i}$ which correspond each to one of the $\lceil n_i/2\rceil$ least significant digits of the input integers $A_i$ and $B_i$ of $P_i$ (resp., to one of the $\lceil n_i/2\rceil$ digits of the output product $C_i$ form the $(\lceil n_i/2\rceil+1)$-least significant to the $n_i$-least significant).  Finally, we define $\subpInput{} = \cup_{i=1}^{\nu(n')} \subpInput{}_i$ and $\subpOutput{} = \cup_{i=1}^{\nu(n')} \subpOutput{}_i$.

For each Type 1 $n'$-MSP $P_i$ generated by $\mathcal{A}$, with input integers of  size $n_i\geq n'$, we denote as $T_i$ the set of variables whose values correspond to the at least $\lceil n_i^2/4\rceil$ elementary products  $A_i[j]B_i[k]$ for $j,k = 0,1,\ldots,\left(n_i-1\right)/2$ and $n_i/4\leq j+k\leq 3n_i/4$. By definition, these elementary products are computed by any integer multiplication algorithm from the standard class towards computing the output digits from the $\left(\halfsubpsize{}+1\right)$-least significant up to the $\subpsize$-th least significant one (i.e., the values corresponding to vertices in $\subpOutput{}_i$). Further, we denote as $\mathcal{T}_i$ the set of vertices corresponding to the variables in $T_i$, and we define $\mathcal{T}= \cup_{i=1}^{\nu_1(n')} \mathcal{T}_i$.  Clearly, $|\mathcal{T}_i|\geq n_i^2/4$, and $4|\mathcal{T}|= \sum_{i=1}^{\nu_1(8M)} n_i^2$.

At least one of the $\nproc{}$ processors must compute (a) exactly $\mathcal{T}/P$ distinct values corresponding to vertices in $\mathcal{T}$ (henceforth denoted as $\mathcal{T}'$) or (b) $n'\nu(n')/P$ distinct values corresponding to vertices in $\subpOutput{}$ (denoted as $\subpOutput{}'$). Let $P'$ be such processor, we denote as $X'$ the set of vertices composed of the vertices corresponding to the at most $\alpha n/\nproc{}$ input values initially stored in the local memory of $\nproc{}'$. Finally, we denote as $X''$ the set of vertices of $G^{\mathcal{A}}$ which are descendants \emph{only} of vertices in  $X'$. That is, there is no directed path connecting vertices corresponding to the input values not in $X'$ to vertices in $X''$. 
Intuitively, the values corresponding to the vertices in  $X'\cup X''$ are all those which are either stored in $P'$ at the beginning of the computation ($X'$) or those which can be computed by $\nproc{}'$ without requiring any value from any other processor ($X''$).

Consider the set $L$ of vertices of $G_{\mathcal{A}}$ corresponding to the  memory words sent or received  by $P'$ executed during the computation of $\mathcal{A}$ (i.e., $|L|=g'$).
Below we show that $g'\geq c \min\{n',\sqrt{|\mathcal{T'}|}\}=c \min\{n',\sqrt{|\mathcal{T}|/P}$ for case (a) and $g'\geq \min\{n',|\mathcal{Z}'|\}= \min\{n',|\mathcal{Z}|/\nproc{}\geq n'\nu(n')\}$ for case (b), from which the theorem follows. We consider the two cases separately: \\

\noindent\textbf{Case (a):} For each Type 1 $n'$-MSP $P_i$  let $\mathcal{T}'_i$ denote the subset of $\mathcal{T}'$ composed of the vertices in the sub-CDAg corresponding to $P_i$. As the $\nu_1$ sub-CDAGs corresponding each to one of the Type 1 $M$-MSPs are vertex-disjoint, so are the sets $\mathcal{T}'_i$. Hence, the $\mathcal{T}'_i$'s constitute a partition of $\mathcal{T}'$. 
    Let $A_i$ and $B_i$ (resp., $C_i$) denote the input integers (resp., output product) of $P_i$, where $A_i$ and $B_i$ can be expressed in base-$s$ using up to $n_i$ digits.
    
    For $k = n_i/4,n_i/4+1,\ldots,3n_i/4-1$, we say that $\nproc{}'$ ``\emph{contributes to}  $C_i[k]$'' if \emph{any} of the at least $n_i/4$ elementary products $A_i[r]B_i[s]$, for $0\leq r,s< n_i/2$ and $r+s=k$, correspond to any of the vertices in $\mathcal{T}'_i$. Further, we say that a $A_i[r]$ (resp., $B_i[r]$) is ``\emph{accessed by} $\nproc{}'$'' if \emph{any} of the elementary multiplications $A_i[r]B_i[l]$ (resp., $A_i[l]B_i[r]$), for $l= 0,1,\ldots,n_i-1$, correspond to any of the vertices in $\mathcal{T}^{(j)}_i$. In the following we denote as $\alpha_i$ (resp., $\beta_i$) as the number of digits of $A_i$ (resp., $B_i$) which are accessed during $\mathcal{C}_j$, and let $\gamma_i$ denote the number of digits of $C_i$ to whom $\nproc{}'$ contributes to. 
    By definition, $\max\{\alpha_i,\beta_i\}\geq \lceil\sqrt{|\mathcal{T}_i|}\rceil$ and $\gamma_i\geq \lceil |\mathcal{T}_i|/\max\{\alpha_i,\beta_i\}\rceil$.
    
Assume that for at least one of the $\nu_1(8M)$ Type 1 $n'$-MSPs, henceforth denoted as $P_i$, at least $\min\{n'/4, \sqrt{|\mathcal{T}|/2P}\}$ among the $n_i/2$ least significant digits of $A_i$ or $B_i$ are accessed during $\mathcal{C}_j$. We denote as $Y$ the set of vertices of $G^{\mathcal{A}}$ corresponding to such values. Clearly $Y\subseteq\mathcal{Y}_i\subseteq \mathcal{Y}$. In order for $\nproc{}'$ to access the at least $\min\{n'/4, \sqrt{|\mathcal{T}|/2P}\}$ values corresponding to vertices in $Y$ there must be no path connecting any of the at least $n-\alpha n/P$ input vertices of $G^{\mathcal{A}}$ corresponding to input values not initially stored in the local memory of $P'$ to vertices in $Y$ which does not have at least one vertex in $L$. That is,  $L$ has to be a \emph{dominator set} of $Y$ with respect to the at least $n-\alpha n/P$ input vertices of $G^{\mathcal{A}}$ corresponding to values not initially stored in the local memory of $P'$.  
    
    By Lemma~\ref{lem:domtype3}, any denominator $D$ of $Y$ satisfies
  $$|D|\geq |Y|\geq \min\{n'/4, \sqrt{|\mathcal{T}|/2P}\}.$$   Hence, there are $\min\{n'/4, \sqrt{|\mathcal{T}|/2P}\}$ vertex-disjoint paths connecting vertices input vertices of $G^{\mathcal{A}}$ to vertices of $\mathcal{Z}''$. By definition, at most $\alpha n/P$ of these paths may share vertices with $X'\cup X''$. In turn this implies that $$|L|=g'\geq \min\{\frac{n'}{4}, \sqrt{\frac{|\mathcal{T}|}{2\nproc{}}}\} - \alpha \frac{n}{\nproc{}}\geq \min\{\frac{n'}{4}, n'\sqrt{\frac{\nu_i(n')}{2P}}\} - \alpha \frac{n}{\nproc{}}.$$
  
  Assume instead that for all Type 1 $n'$-MSPs $P_i$ generated by $\mathcal{A}$ strictly less than  $\min\{n'/4, \sqrt{|\mathcal{T}|/2P}\}$ among the $n_i/2$ least significant digits of $A_i$ or $B_i$ are accessed during $\mathcal{C}_j$. That is $\max_{i\in\{1,2,\ldots,\nu_1(n')\}}\{\alpha_i,\beta_i\}<2M$. Let $C_i[k]$ be active during $\mathcal{C}_j$. In order for $\nproc{}'$ to compute $C_i[k]$ \emph{entirely} (i.e., without receiving partial accumulators of the summation $\sum_{r,s\geq 0|r+s=k} A_i[r]B[s]$ from other processors and/or sending partial accumulators for the summation to other processors by means of an \io{} operation), it will be necessary to evaluate all the at least $n_i/4$ elementary products which are added in the summation during $\mathcal{C}_j$ itself. Thus, at least $n_i/4\geq n'/4$ values corresponding to digits form $A_i$ and $B_i$ must be accessed during $\mathcal{C}_j$. As, by assumption, for all $i\in\{1,2,\ldots,\nu_1(n')\}$ we have $\alpha_i,\beta_i<\min\{n'/4, \sqrt{|\mathcal{T}|/2P}\}$,  none of the values $C_i[k]$ which are active during $\mathcal{C}_j$ are computed entirely during $\mathcal{C}_j$ itself.

  As $\mathcal{C}$ is a parsimonious computation, for any value $C[k]$ to the computation of whom $\nproc{}'$ contributed either $\nproc{}'$ will receive some partial accumulator computed by at least another processor, or the value of  partial accumulator for $C[k]$ computed by $\nproc{}'$ must be sent to at least another processor. We assume that such accumulators are received or sent a single memory word for any value $C[k]$ to whom $\nproc{}'$ contributes. Thus for all  the $\sum_{i=1}^{\nu_1}\gamma_i$ values $\nproc{}'$ contributes to partial accumulators must be either sent from $\nproc{}'$ to other processors, or received by $\nproc{}'$ from another processor by means of an \io{} operation. As, by assumption $\max_{i\in\{1,2,\ldots,\nu_1(n')\}}\{\alpha_i,\beta_i\}<\sqrt{|\mathcal{T}|/2P}$, form the previous considerations we have:
  \begin{align*}
      \sum_{i=1}^{\nu_1(n')}\gamma_i &\geq \sum_{i=1}^{\nu_1(n')}\bigl\lceil\frac{|\mathcal{T}'_i|}{\max\{\alpha_i,\beta_i\}}\bigr\rceil\\
      &> \sum_{i=1}^{\nu_1(n')}\frac{|\mathcal{T}'_i|}{\sqrt{\frac{|\mathcal{T}|}{2\nproc{}}}}\\
      &=\sqrt{\frac{2\nproc{}}{|\mathcal{T}|}}\sum_{i=1}^{\nu_1(n')}|\mathcal{T}'_i|\\
      &=\sqrt{\frac{2\nproc{}}{|\mathcal{T}|}}|\mathcal{T}'|\\\
      &=\sqrt{\frac{2\nproc{}}{|\mathcal{T}|}}\frac{|\mathcal{T}|}{P}\\
      &\geq \sqrt{\frac{n'^2\nu_1(n')}{2\nproc{}}}.
  \end{align*}

\noindent\textbf{Case (b):} If $|\mathcal{Z}'|>n'/2$, let $\mathcal{Z}''$ be any subset of $\mathcal{Z}'$ such that $||\mathcal{Z}'|=n'/2$. Otherwise, $|\mathcal{Z}''|= |\mathcal{Z}'|$. In order for $\nproc{}'$ to compute the values corresponding  to the vertices in  $\subpOutput{}''$ there must be no path connecting any of the at least $n-\alpha n/P$ input vertices of $G^{\mathcal{A}}$ corresponding to input values not initially stored in the local memory of $P'$ to a vertex in $\mathcal{Z}''$ which does not have at least one vertex in $L$. That is,  $L$ has to be a \emph{dominator set} of $\subpOutput{}''$ with respect to the at least $n-\alpha n/P$ input vertices of $G^{\mathcal{A}}$ corresponding to values not initially stored in the local memory of $P'$. 

From Lemma~\ref{lem:domtype1}, any dominator set $D$ of $\mathcal{Z}''$ satisfies $|D|\geq |\mathcal{Z}'|/2\geq \min\{n',n'\nu(n')/P\}$. Hence, there are $\min\{n',n'\nu(n')/P\}$ vertex-disjoint paths connecting vertices input vertices of $G^{\mathcal{A}}$ to vertices of $\mathcal{Z}''$. By definition, at most $\alpha n/P$ of these paths may share vertices with $X'\cup X''$. In turn this implies that $|L|=g'\geq \min\{n',n'\nu(n')/P\} - \alpha n/\nproc{}$.\\

Clearly, 
$$
\min\{\sqrt{\frac{n'^2\nu_1(n')}{2\nproc{}}}, \frac{n'\nu(n')}{\nproc{}}\}\geq\frac{1}{2} \left(\sqrt{\frac{n'^2\nu_1(n')}{2\nproc{}}}+\frac{n'\nu(n')}{\nproc{}}\right).
$$
The statement follows.
\end{proof}
\end{document}